\documentclass[draftclsnofoot,onecolumn,12pt]{IEEEtran}

\usepackage{bbm}
\usepackage{amsmath,graphics,amssymb,epsfig,subfigure,color,amsthm,cite}
\usepackage{array}
\usepackage{multirow}
\usepackage{enumerate}
\usepackage[ruled]{algorithm2e}
\SetKwInOut{Input}{Input}
\SetKwInOut{Output}{Output}
\DontPrintSemicolon
\usepackage{amsmath}

\DeclareMathOperator*{\argmin}{arg\,min}

\usepackage{bm}
\usepackage[cmintegrals]{newtxmath}
\usepackage{tablefootnote}
\usepackage{amsmath}
\usepackage{algorithmic,algorithm2e}
\usepackage{mathtools}
\usepackage{cuted}
\usepackage{flushend}
\usepackage{color}
\usepackage{caption}
\usepackage{tabularx,booktabs}
\usepackage[flushleft]{threeparttable}

\newcommand{\rom}[1]{\uppercase\expandafter{\romannumeral #1 \relax}}

\newtheorem{thm}{Theorem}

\newtheorem{cor}{Corollary}

\begin{document}
\title{Bayesian Federated Learning \\ over Wireless Networks}
 
\author{ Seunghoon Lee, Chanho Park, Song-Nam Hong,  
\\Yonina C. Eldar, and Namyoon Lee 
\thanks{ S. Lee, C. Park, and N. Lee are with the Department of Electrical Engineering, POSTECH, Pohang, Gyeongbuk, 37673 Korea (e-mail: \{shlee14 chanho26, nylee\}@postech.ac.kr).} 
\thanks{S.-N. Hong is with Department of Electrical Engineering, Hanyang University, Korea (e-mail: snhong@hanyang.ac.kr).}
\thanks{Y. C. Eldar is with the Math and CS Faculty, Weizmann Institute of Science,
Rehovot, Israel (e-mail: yonina.eldar@weizmann.ac.il).}}
\vspace{-2mm}

\maketitle
\setlength\arraycolsep{2pt}
\makeatletter
\newcommand{\vast}{\bBigg@{3.5}}
\newcommand{\Vast}{\bBigg@{4.5}}
\makeatother
\vspace{-12mm}

\begin{abstract}
Federated learning is a privacy-preserving and distributed training method using heterogeneous data sets stored at local devices. Federated learning over wireless networks requires aggregating locally computed gradients at a server where the mobile devices send statistically distinct gradient information over heterogenous communication links. This paper proposes a Bayesian federated learning (BFL) algorithm to aggregate the heterogeneous quantized gradient information optimally in the sense of minimizing the mean-squared error (MSE). The idea of BFL is to aggregate the one-bit quantized local gradients at the server by jointly exploiting i) the prior distributions of the local gradients, ii) the gradient quantizer function, and iii) channel distributions. Implementing BFL requires high communication and computational costs as the number of mobile devices increases. To address this challenge, we also present an efficient modified BFL algorithm called scalable-BFL (SBFL).  In SBFL, we assume a simplified distribution on the local gradient. Each mobile device sends its one-bit quantized local gradient together with two scalar parameters representing this distribution. The server then aggregates the noisy and faded quantized gradients to minimize the MSE.  We provide a convergence analysis of SBFL for a class of non-convex loss functions. Our analysis elucidates how the parameters of communication channels and the gradient priors affect convergence. From simulations, we demonstrate that SBFL considerably outperforms the conventional sign stochastic gradient descent algorithm when training and testing neural networks using MNIST data sets over heterogeneous wireless networks.

\end{abstract}

\section{Introduction}
Federated learning is a decentralized approach to train machine learning models at a server using distributed and heterogeneous training data sets placed at mobile devices, without sharing raw data with a server \cite{mcmahan2017communication, konevcny2016federated, bonawitz2019towards}.  This machine learning approach has received considerable attention from the research community and industries because of the myriad of applications that require the privacy of user-generated data, such as activity on mobile phones. Federated learning performs distributed model training iteratively with two operations: 1) model optimization with local data sets, and 2) model aggregation, i.e., model averaging \cite{yang2019federated}. In every round, the server sends a global model to a set of available mobile devices. Each device optimizes the model with locally available data and then sends its updated model parameters (or updated local gradient) to the server via communication links. The server updates the global model by averaging the local models or gradients received from the mobile devices and shares it in the next iteration.  



Optimizing a global model using local data sets placed at a massive number of mobile devices is a challenging task \cite{dean2008mapreduce, li2014communication}. The primary challenge is the high communication cost when updating the local computations from the mobile devices and the server.  The communication cost is proportional to both the size of the global model parameters and the number of mobile devices connected to the server. The model size is typically tens of millions when training complex neural networks \cite{liu2017survey}, and the number of mobile devices can be a few hundred and more depending on applications \cite{yang2019federated}. When sending locally computed model parameters simultaneously over wireless links, the server cannot decode them all successfully under a limited bandwidth constraint.   


Lossy compression of local gradients is a practical solution to reduce the size of the message per communication round.  Compression can be performed via sparsification, or quantization \cite{alistarh2017qsgd, wangni2018gradient, alistarh2018convergence, dai2019hyper, bernstein2018signsgd, shi2019distributed, shi2019understanding, horvath2019stochastic,Eldar_vecQ}. One popular approach of  gradient sparsification is to send the top-$k$ magnitude coordinates of a local gradient vector \cite{shi2019understanding, shi2019distributed}.  An alternative approach is to use  optimal vector quantization methods that minimize quantization error \cite{dai2019hyper,Eldar_vecQ}. Such approaches may be challenging to implement in practice. Scalar quantization techniques are very popular due to their simplicity \cite{alistarh2017qsgd,bernstein2018signsgd}. In particular, one-bit gradient quantization with sign stochastic gradient descent (signSGD) \cite{bernstein2018signsgd} has been studied extensively in the recent literature. 

Notwithstanding significant progress in improving communication efficiency by compression, prior work did not consider aggregation of local gradients by jointly embracing heterogeneous local data distributions, varying communication link reliabilities per device, and quantization effects. Most prior works mainly focus on a separate communication and learning system design approach -- decoding the local gradients and aggregating them independently. Under the premise that the local gradients sent by the mobile devices are perfectly decodable at the server using proper modulation and coding techniques per link, simple gradient aggregation methods (e.g., gradient-averaging or model-averaging) have been considered. For example, when receiving faded and noisy versions of one-bit local gradients from mobile devices, the most popular method to aggregate the signs of the gradients is simple majority voting in signSGD \cite{bernstein2018signsgd, bernstein2018majority, sattler2019robust}. One drawback of this strategy is that it fails to achieve the optimal performance in detecting the sign of the sum of one-bit gradients when the distributions of the local gradients and the communication links' qualities are not identical across mobile devices.  Another drawback of this separate design approach is that the server must perform at least as many decoding operations as the number of devices sending the local gradients before aggregation, which may give rise to a significant delay in federated learning systems.  
 
Over-the-air aggregation is a new paradigm to jointly design the communication and learning system by harnessing the superposition nature of the wireless medium \cite{gunduz2019machine, amiri2020federated, chen2019joint, yang2020federated,Eldar_FL, zhu2020one}. Although this approach opens new opportunities in designing the wireless and learning system jointly, it is limited to time division duplex wireless systems, in which channel state information is available at the mobile devices for uplink transmissions. In current frequency-division-duplexing LTE systems,  orthogonalized multiple access such as narrowband IoT may be preferred in implementing federated learning systems \cite{ratasuk2016overview, chen2017narrow}, where each mobile device performs uplink transmission using orthogonal subcarriers.

In this work, we consider a federated learning system, in which mobile devices send  local gradients to the server per communication round using orthogonalized multiple access channels, each with heterogeneous wireless link quality. We suggest performing  joint detection and aggregation of the local gradients using a Bayesian framework.  Our key contributions are summarized as follows.

\begin{itemize}
    \item We introduce a Bayesian approach for federated learning called BFL. The key innovation of the proposed algorithm is to jointly exploit 1) the joint distribution of local gradients, 2) the one-bit quantizer function for gradient compression, and 3) the distributions of communication channels when aggregating the quantized local gradients. Under the premise that the prior distributions of the local gradients are jointly Gaussian, this aggregation method is optimal in the sense of minimizing the mean-squared error (MSE) per communication round. In practice, the MSE-optimal aggregation method may be difficult to implement due to the lack of knowledge about the joint distribution of local gradients and additional communication costs. 
    
    \item We present an efficient algorithm referred to as scalable-BFL (SBFL) to resolve the practical challenges in BFL. In SBFL we use 1) a simple prior distribution model for the local gradient and 2) a sub-optimal Bayesian aggregation function.  Specifically, we model each local gradient vector's prior distribution as a simple independent identically distributed (iid) Gaussian or Laplacian in which the mean and variances are estimated in every round. Instead of finding the MSE-optimal aggregation function that requires knowledge of joint distribution of the priors and involves multi-dimensional integrations, we use a computationally-efficient Bayesian aggregation function in closed-form, which achieves the minimum MSE under an iid Gaussian prior.  Interestingly, the proposed function can be implemented with a shallow neural network using the hyperbolic tangent activation function. The shallow neural network parameters are determined by the standard deviation of the prior distribution, signal-to-noise ratios (SNRs) of communication links, and channel fadings. Under the assumption of independent Gaussian prior, we also present the Bussgang-based linear minimum MSE (BLMMSE) aggregation method for SBFL and compare it with the proposed nonlinear MMSE aggregation function. 
    
	\item We provide a convergence analysis of the proposed SBFL algorithm for a class of non-convex loss functions. The crucial step in our proof is to show that the proposed Bayesian aggregation function is the unbiased estimator of the sum of sign gradients, and the corresponding MSE is bounded by a constant per communication round.  Then, under mild assumptions, we derive the convergence rate of SBFL to show how fast it decreases a loss function in terms of the number of iterations and the MSE. Our result shows that the SBFL algorithm ensures convergence to a stationary point of a non-convex and smooth loss function.
	
	\item We evaluate the performance of SBFL in two scenarios, including 1) synthetic datasets and  2) MNIST datasets, i.e., non-synthetic datasets. First, we train a linear estimator for which the loss function is strongly convex with both homogeneous and heterogeneous datasets. In such settings, we demonstrate that SBFL results in lower  training loss than signSGD for both datasets. Then, we train a convolutional neural network (CNN) over the MNIST dataset to verify the performance gain in a deep FL setup.  We demonstrate that SBFL improves the accuracy of the trained model when using heterogeneous data sets in wireless networks. The gain of SBFL results from the exploitation of side-information including both the (approximate) prior distribution of local gradients and heterogeneous channel distributions when aggregating the gradient information per communication round. From simulations, we verify that this gain is attainable when training convex and non-convex loss functions with different data settings. 
\end{itemize}

The rest of this paper is organized as follows: Section II provides some background of FL using signSGD. Section III introduces the concept of BFL. Section VI presents SBFL along with its possible generalizations. Convergence analysis of SBFL for optimizing a smooth and non-convex loss function is presented in Section V. Section VI provides simulation results to show the gain of SBFL over signSGD. Finally, we conclude in Section VII and suggest possible future work. Detailed proofs of our main results are provided in the appendix.

    \section{ Wireless Federated Learning}  
    We consider a wireless federated learning system, in which a server (e.g., a base station) is connected to $K$ mobile devices (e.g., mobile phones) via a wireless network.  The primary learning task is to collaboratively train a global model parameter across these distributed mobile devices without uploading local privacy-sensitive data.  Training data sets are separately placed at each wireless device.  Let $\mathcal{X}_k=\left\{{\bf x}_{k}^{i}, {\bf z}_{k}^{i}\right\}_{i=1}^{N_k}$ be the training set of device $k$ for $k\in [K]$, where $N_k=|\mathcal{X}_k|$ is the number of training data points stored at mobile device $k$. Each training data point ${\bf x}_k^{i}$ is drawn from a distribution $P_{{\sf x}_k}$.  Note that the distributions of different devices are distinct. We denote the total number of data points for training by $N=\sum_{k=1}^KN_k$.     We define a loss function at the $k$th device with a shared global model parameter ${\bf w} \in \mathbb{R}^M$ as
    \begin{align}
        f_k({\bf w}) = \frac{1}{N_k} \sum_{i=1}^{N_k}\ell\left({\bf x}_{k}^{i}, {\bf z}_{k}^{i};{\bf w}\right),
    \end{align}
    where $\ell(\cdot): \mathbb{R}^M \times\mathbb{R} \rightarrow \mathbb{R}$ is the loss function of the prediction with the shared model parameter ${\bf w}$ for  training example $\left({\bf x}_k^i,{\bf z}_{k}^{i}\right)$.   This loss function can be either convex or non-convex. The \textit{global loss function} over all distributed training data sets is a sum of local loss functions,
    \begin{align}
        F\left({\bf w}\right)=\frac{1}{N} \sum_{k=1}^{K}\sum_{i=1}^{N_k}\ell\left({\bf x}_{k}^{i}, {\bf z}_{k}^{i};{\bf w}\right)= \sum_{k=1}^{K}\frac{N_k}{N}f_k({\bf w}).
    \end{align}
    
    The federated learning task is to optimize the global model parameter ${\bf w}$ without sending local data points to the server.  We explain how this global model parameter is optimized in a distributed manner using a wireless network. In particular, we focus on signSGD for the distributed optimization because it allows considerably reducing the costs in uplink communications.  Federated learning using signSGD consists of the following five operations:

        {\bf Model sharing via downlink communications:} In the $t$th communication round, the server broadcasts a global model parameter ${\bf w}^t$ to the mobile devices. We assume that this downlink communication is delay and error-free, i.e., all mobile devices decode the message ${\bf w}^t$ perfectly. 
    
    {\bf Model optimization with local data:} Each mobile device computes local gradient $\nabla f_k\left({\bf w}^{t}\right)$ using both its own training data set $\mathcal{X}_k$ for $k\in [K]$ and the model parameter received from the server ${\bf w}^{t}$. We let ${\bf g}_k^{t}=\left[g_{k,1}^t,\ldots, g_{k,M}^t\right]^{\top}=\nabla f_k\left({\bf w}^{t}\right)$ be the local gradient computed at mobile device $k\in [K]$ at communication round $t\in [T]$.  For ease of exposition, we assume a full-batch size for the gradient computation.

            {\bf One-bit compression of the local gradient:} To reduce the communication cost, each mobile device compresses local gradient ${\bf g}_k^{t}$ using the one-bit quantizer by simply taking the sign of ${\bf g}_k^{t}$, i.e., 
    \begin{align}
        {\bf \hat g}_k^t ={\sf sign}\left({\bf g}_k^t\right).
    \end{align}
     

   {\bf Uplink transmission:} Mobile device $k$ for $k\in [K]$ sends one-bit quantized local gradient information ${\bf \hat g}_k^t$ to the server over a faded and noisy communication channel. We mainly consider the orthogonalized massive access channel, in which $K$ local devices send the one-bit gradient parameter to the server simultaneously by using orthogonal subchannels. We focus on the real part of the complex base-band equivalent model for ease of exposition, which can readily be extended using quadrature amplitude modulation signaling to orthogonal BPSK symbols.

    
    Let $h_k^t\in \mathbb{R}$ be the fading channel coefficient of the $k$th subchannel during communication round $t$. This fading channel is assume to be a {\it constant} during the uplink transmission, and change independently every communication round $t$ while remaining constant during the communication round, i.e., a block fading process. Then, the received signal of the server from the $k$th subchannel of the $t$th communication round is
    \begin{align}
        {\bf  y}_{k}^t=h_k^t {\sf sign}\left({\bf g}_{k}^t\right) +{\bf  n}_{k}^t, 
        \label{eq:channel}
    \end{align}
    where ${\bf n}_{k}^t$ is the Gaussian noise of the $k$th subchannel with zero-mean and variance $\sigma^2$, i.e., ${\bf n}_{k}^t\sim \mathcal{N}\left({\bf 0},\sigma^2_k{\bf I}\right)$. From \eqref{eq:channel}, the likelihood of ${\bf y}_{k}^t$ given ${\bf g}_{k}^t$ and assuming fixed $h_k^t$ is
    \begin{align}
        P\left({\bf y}_{k}^t|{\bf g}_{k}^t \right)=\prod_{m=1}^M\frac{1}{\sqrt{2\pi \sigma^2}}\exp\left(-\frac{\left|y_{k,m}^t-h_k^t{\sf sign}\left({g}_{k,m}^t\right)\right|^2}{2\sigma^2_k}\right). \label{eq:likelihood}
    \end{align}

  {\bf Majority-voting for aggregation and model update:} The server performs maximum-likelihood detection (MLD) to estimate one-bit local gradients per subchannel. Specifically, using the received signal of the $k$th subchannel, one-bit local gradient information is estimated as
  \begin{align}
  	 {\bf \hat g}_k = {\sf sign}\left( \frac{{\bf y}_k^t}{h_k^t}\right),
  \end{align} 
for $k\in [K]$. Then, to attain a global gradient estimate, the server applies a majority-vote based aggregation method, which takes the sign of the sum of the estimated one-bit local gradients:
  \begin{align}
  	 {\bf \hat g}_{\sf MV} = {\sf sign}\left(\sum_{k=1}^K{\bf \hat g}_k \right)={\sf sign}\left(\sum_{k=1}^K{\sf sign}\left( \frac{{\bf y}_k^t}{h_k^t}\right)\right).
  \end{align} 
  After computing ${\bf \hat g}_{\sf MV}$, the server updates the model parameter as    \begin{align}
        {\bf w}^{t+1}&= {\bf w}^t -\gamma^t {\bf \hat g}_{\sf MV}^t,
    \end{align}
    where $\gamma^t\in(0,1)$ is a learning rate of the gradient descent algorithm. The updated model parameter is then sent to the mobile devices and the learning procedure moves to the next round.  To speed up convergence, the server may update the model parameter by taking a weighted average between the currently estimated gradient and the previously estimated gradients ${\bf m}^{t}=\delta^t {\bf m}^{t-1}+{\bf \hat g}^{t}_{\sf MV}$ as
     \begin{align}
     {\bf w}^{t+1}= {\bf w}^{t} -\gamma^t {\bf m}^{t},
    \end{align}
     where $\delta^t\in(0,1)$ is a weight parameter and ${\bf m}^{0}$ is the initial value of the moment.

\section{ Bayesian Federated Learning}
In this section, we present a framework for BFL.  The motivation of BFL is to find an optimal aggregation method of local gradients from a Bayesian viewpoint.  The majority-voting based aggregation is effective in detecting the sign of the sum of the one-bit local gradients when the SNRs and  fadings are identical across mobile devices. In heterogenous wireless networks, however, this majority-voting based aggregation cannot be optimal in detecting the sign of the local gradients' sum due to the distinct reliabilities of the communication links. In addition, to enhance the FL performance, the server may require to accurately know the sum of local gradients rather than the sum of the local gradients' signs.  Motivated by these facts, we aim at finding an optimal aggregation method of the local gradients in the sense of minimizing MSE. The key idea of BFL is to aggregate local gradients per iteration by jointly harnessing the knowledge of 1) the prior distributions of local gradients, 2) the gradient quantizer function,  and 3) the likelihood functions of the communication channels.



 \subsection{Algorithm}

    {\bf Computation of local gradient prior:} Unlike conventional federated learning, BFL exploits the local gradient's prior information. Let $P({\bf g}_{k}^t)$ be the prior distribution of  ${\bf g}_{k}^t$ computed by mobile device $k$ at communication round $t$. Characterizing the prior distribution is very challenging because it not only depends on the loss function $f_k({\bf w}^t)$ but also on the underlying distribution of data samples; the exact characterization of the prior distribution is impossible in general. To overcome this difficulty, we model the distribution of ${\bf g}_k^{t}$ as a Gaussian prior distribution with a proper moment matching technique. Specifically, we assume the prior is a multivariate normal distribution with mean vector ${\bm \mu}_k^t=\mathbb{E}\left[ {\bf g}_{k}^t\right]   $ and covariance matrix ${\bm \Sigma}_k^t =\mathbb{E}\left[ \left({\bf g}_{k}^t-{\bm \mu}_k^t \right)\left({\bf g}_{k}^t-{\bm \mu}_k^t\right)^{\top}\right]$ as
    \begin{align}
        P\left( {\bf g}_{k}^t \right)=\frac{1}{(2\pi)^{\frac{M}{2}} \det({\bm \Sigma}_k^t)} e^{-\frac{1}{2}\left({\bf g}_k^t -{\bm \mu}_k^t\right)^{\top}\left({\bm \Sigma}_k^t\right)^{-1} \left({\bf g}_k^t -{\bm \mu}_k^t\right)}.
        \label{eq:prior}
    \end{align}
   The prior distribution can change over communication rounds according to the underlying distribution of data samples and the local loss function.

  Notice that local gradients ${\bf g}_k^t=\nabla f_k\left({\bf w}^{t}\right)$ and ${\bf g}_{\ell}^t=\nabla f_{\ell}\left({\bf w}^{t}\right)$ for $k\neq \ell$ are statistically correlated because they are evaluated at a common model parameter ${\bf w}^t$. To illustrate this, consider a linear regression loss function 
  \begin{align}
  	f_k({\bf w}^t)=\|{\bf X}_k^{\top}{\bf w}^t-{\bf y}_k\|_2^2,
  \end{align}
   where ${\bf X}_k=\left[ {\bf x}^{1}_k, {\bf x}^{2}_k,\ldots, {\bf x}^{N_k}_k \right]\in \mathbb{R}^{M\times N_k}$ and ${\bf z}_k=\left[z_k^1, \ldots, z_k^{N_k}\right]^{\top}\in \mathbb{R}^{N_k\times 1}$ for $k\in [K]$. Furthermore, we assume all local data distributions $P({\bf X}_k, {\bf z}_k)$ for $({\bf X}_k,{\bf z}_k)$ are statistically independent, each with $\mathbb{E}\left[{\bf X}_k^{\top}{\bf X}_k\right]={\bf R}_k$ and $\mathbb{E}\left[ {\bf X}_k^{\top}{\bf z}_k\right]={\bf 0}$. In this case, the local gradient at mobile device $k\in[K]$ computed using model parameter ${\bf w}^t$ is ${\bf g}_k^t = {\bf X}_k^{\top}{\bf X}_k {\bf w}^t - {\bf X}_k^{\top}{\bf zz}_k$. The correlation matrix between ${\bf g}_k^t$ and ${\bf g}_{\ell}^t$ conditioned on ${\bf w}^t$ is computed as
\begin{align}
 \mathbb{E}\left[{\bf g}_k^t\left({\bf g}_{\ell}^t\right)^{\top}\mid {\bf w}^t\right]&= \mathbb{E}\left[ {\bf X}_k^{\top}{\bf X}_k {\bf w}^t\left( {\bf w}^t\right)^{\top}\left({\bf X}_{\ell}^{\top}{\bf X}_{\ell}\right)^{\top}\right] \nonumber\\
 &= \mathbb{E}\left[ {\bf X}_k^{\top}{\bf X}_k\right] {\bf w}^t\left( {\bf w}^t\right)^{\top} \mathbb{E}\left[\left({\bf X}_{\ell}^{\top}{\bf X}_{\ell}\right)^{\top}\right]\nonumber\\
 &= {\bf R}_k {\bf w}^t\left( {\bf w}^t\right)^{\top}{\bf R}_{\ell}.
	\end{align}z The correlation structure between ${\bf g}_k^t$ and ${\bf g}_{\ell}^t$ can be more complicated, if the data distributions $P({\bf X}_k, {\bf z}_k)$ and $P({\bf X}_{\ell}, {\bf z}_{\ell})$ are also correlated across mobile devices.

    {\bf One-bit compression of the local gradient:} To reduce the communication cost, each mobile device compresses its local gradient ${\bf g}_k^{t}$ using a one-bit quantizer. To make the one-bit quantized output uniformly distributed, the mobile device first performs zero-mean normalization, i.e.,	
    \begin{align}
    	{\bf \bar g}_{k}^t ={\bf g}_{k}^t - {\bm \mu}_k^t.
    \end{align}
The normalized gradient is then quantized by  taking the sign of ${\bf \bar g}_k^{t}$, i.e., 
    \begin{align}
        {\bf \tilde g}_k^t ={\sf sign}\left({\bf \bar g}_k^t\right).
    \end{align}
    Each mobile devices sends both this binary gradient information $ {\bf \tilde g}_k^t$ with the parameters of the prior distribution, ${\bm \mu}_k^t$ and ${\bm \Sigma}_k^t$, as additional side-information. In this section, we assume that ${\bm \mu}_k^t$ and ${\bm \Sigma}_k^t$ are delivered to the server perfectly. 
    

%
%
%
%
    


    {\bf Bayesian aggregation of local gradients:}  In the $t$th communication round, the server updates the parameter using channel outputs $\left\{{\bf y}_{k}^{t}\right\}$ for $k\in[K]$. We present the optimal gradient aggregation function from a Bayesian perspective.  The server has knowledge about the channel distributions $P\left({\bf y}_{k}^t \mid {\bf \bar g}_{k}^t\right)$ and the marginal distributions of local gradients $P\left( {\bf \bar g}_{k}^t\right)$ for $k\in [K]$. In BFL, we assume that the server knows a joint distribution of local gradients, i.e.,  $P\left(  {\bf \bar g}_{1}^t, \ldots, {\bf \bar g}_{K}^t\right)$. This is a genie-aided assumption to derive the MSE-optimal aggregation function.
    
    Let $\triangleq U({\bf y}_1^t,\ldots, {\bf y}_K^t):\mathbb{R}^{MK}\rightarrow\mathbb{R}^{M}$ be an aggregation function to estimate the sum of local gradients ${\bf \bar g}^t_{\Sigma}=\sum_{k=1}^K{\bf \bar g}^t_{k}$.    Our goal is to find an aggregation function $U\left({\bf y}_1^t,\ldots, {\bf y}_K^t\right)$ such that
  \begin{align}
  	{\bf \hat g}^t_{\sf MMSE} =\argmin_{U:\mathbb{R}^{MK}\rightarrow\mathbb{R}^{M}}  \mathbb{E}\left[\left\|{\bf   \bar g}_{\Sigma}^t-U({\bf y}_1^t,\ldots, {\bf y}_K^t)\right\|_2^2\right] \label{eq:opti}.
  \end{align}
The MSE-optimal aggregation function for the optimization problem in \eqref{eq:opti} is obtained by the conditional expectation \cite{stark1986probability}:
\begin{align}
	{\bf\hat g}_{{\sf MMSE}}^t= \sum_{k=1}^K{\bm \mu}_{k}^t +\mathbb{E}\left[	{\bf  \bar g}_{k}^t | {\bf y}_1^t, \ldots, {\bf y}_K^t\right], \label{eq:mmse}
\end{align}    
where
   \begin{align}
   	 \mathbb{E}\left[	{\bf   \bar g}_{k}^t | {\bf y}_1^t, \ldots, {\bf y}_K^t\right]=  \int_{-\infty}^{\infty}{\bf \bar g}_{k}^t \frac{P\left({\bf y}_{1}^t, \ldots, {\bf y}_{K}^t\mid {\bf \bar g}_{k}^t\right)}{P\left({\bf y}_{1}^t, \ldots, {\bf y}_{K}^t\right) } P\left({\bf \bar g}_{k}^t \right) {\rm d}{\bf \bar g}_k^t. \label{eq:conditional-mean}
   \end{align}
 Using the chain rule and the independence between ${\sf  y}_{\ell}^t$  for $\ell\in [K]$ conditioned on ${\bf \bar g}_k^t$, the conditional distribution in \eqref{eq:conditional-mean} can be factorized as
    \begin{align}
   	P\left({\bf y}_{1}^t, \ldots, {\bf y}_{K}^t\mid {\bf \bar g}_{k}^t\right) 
   	&=  P\left({\bf y}_{k}^t \mid {\bf \bar g}_{k}^t\right) \prod_{i\neq k}^K P\left({\bf y}_{i}^t \mid {\bf \bar g}_{k}^t\right)  \nonumber\\
   	&=P\left({\bf y}_{k}^t \mid {\bf \bar g}_{k}^t\right) \prod_{i\neq k}^K\int_{-\infty}^{\infty}P\left({\bf y}_{i}^t \mid {\bf \bar g}_{\ell}^t\right)  P\left({\bf \bar g}_{\ell}^t \mid   {\bf \bar g}_{k}^t\right)  {\rm d} {\bf \bar g}_{\ell}^t,\label{eq:posterior}
   	   \end{align}
    where the last equality follows from the fact that ${\bf \bar g}_{k}^t\rightarrow {\bf \bar g}_{\ell}^t \rightarrow {\bf  y}_{\ell}^t$ forms a Markov chain and $ P\left({\bf \bar g}_{\ell}^t \mid   {\bf \bar g}_{k}^t\right)$ for $\ell\in [K]/[k]$ is obtained from the marginalization of joint distribution $P\left(  {\bf \bar g}_{1}^t, \ldots, {\bf \bar g}_{K}^t\right)$. Plugging \eqref{eq:posterior} into \eqref{eq:conditional-mean} leads to the optimal aggregation function. Using ${\bf\hat g}_{\sf MMSE}^t$ in \eqref{eq:mmse}, the server updates the model parameter as    \begin{align}
        {\bf w}^{t+1}&= {\bf w}^t -\gamma {\bf\hat g}_{\sf MMSE}^t.
    \end{align}

    \subsection{Limitations}
  Under the correlated Gaussian priors, there are several critical issues that hinder the use of BFL in practice as specified in the sequel.

    
 {\bf Limited knowledge about the prior distribution:} Since each mobile device sends only ${\bm \mu}_k^t$ and ${\bm \Sigma}_k^t$ to the server, it acquires information only on the marginal distributions of local gradients, i.e., $P({\bf \bar g}_k)$ for $k\in [K]$.  However, to compute $\mathbb{E}\left[	{\bf \bar g}_{k}^t | {\bf y}_1^t, \ldots, {\bf y}_K^t\right]$ for $k\in [K]$, the server must also know conditional distributions $P\left({\bf \bar g}_{\ell}^t \mid   {\bf \bar g}_{k}^t\right) $ for $\ell\neq k$ and $\ell,k\in [K]$. 
 
  {\bf Computational complexity:} One possible approach to resolve this issue is to ignore the correlation among all local gradients. By treating them as statistically independent Gaussian random variables, the server can use the MSE-optimal aggregation function as
        \begin{align}
        	{\bf \hat g}_{\Sigma'}^t=\sum_{k=1}^K{\bm \mu}_{k}^t + \mathbb{E}\left[	{\bf \bar g}_{k}^t | {\bf y}_k^t\right],\label{eq:subMMSE}
        \end{align}
        where 
        \begin{align}
        	\mathbb{E}\left[	{\bf \bar g}_{k}^t | {\bf y}_k^t\right]&=  \int_{-\infty}^{\infty}{\bf \bar g}_{k}^t \frac{P\left(  {\bf y}_{k}^t\mid {\bf \bar g}_{k}^t\right)}{P\left(  {\bf y}_{k}^t\right) } P\left({\bf \bar g}_{k}^t \right) {\rm d}{\bf \bar g}_k^t.
        \end{align}
Computing $\mathbb{E}\left[	{\bf \bar g}_{k}^t | {\bf y}_k^t\right]$, however, also requires an $M$-dimensional integration, which gives rise to a high computational complexity as the model size, $M$, increases.     
        
                 
  {\bf Communication cost:} The communication cost is also a significant bottleneck. The server requires additional information from all mobile devices about parameters, i.e., ${\bm \mu}_k^t\in  \mathbb{R}^M$ and ${\bm {\Sigma}}_k^t\in \mathbb{R}^{M\times M}$ per communication round. When the model size is large, sending such additional information with the one-bit gradients increases the communication cost considerably.
  
  {\bf Accuracy of model parameter estimation:} The Gaussian approximation of the prior distribution can be inaccurate for some loss functions and local data distributions. Furthermore, accurately estimating ${\bm \mu}_k^t\in  \mathbb{R}^M$ and ${\bm \Sigma}_k^t\in \mathbb{R}^{M\times M}$ is another challenge because it can change over every communication round.  When estimating the parameters inaccurately, the prior distribution mismatch effect might cause a performance loss, which eventually can degrade the learning performance.

   
\section{Scalable Bayesian Federated Learning}

This section presents a computation-and-communication efficient BFL algorithm called SBFL. SBFL is scalable to the number of mobile devices and robust to the heterogeneities of both local data distributions and communication link qualities.

    \subsection{Algorithm}
      To make the learning algorithm computation-and-communication efficient, we simplify the local gradient prior and parameterize it by two scalars. Then, we use an element-wise Bayesian aggregation function by assuming all elements of each local gradient are distributed as iid Gaussian random variables. 
    
    
    {\bf Simplification of the local gradient prior:}  We first model all elements of ${\bf g}_k^t$ as iid Gaussian with common mean $\mu_k^t$ and variance $\left(\nu_k^t\right)^2$, i.e.,
    \begin{align}
    		P\left( {\bf g}_{k}^t \right)\simeq \prod_{m=1}^M\frac{1}{\sqrt{2\pi}\nu_k^t}\exp\left(-\frac{|g_{k,m}^t-\mu_k^t|^2}{2(\nu_k^t)^2}\right). \label{eq:iidprior} 
    \end{align}
Each mobile device then estimates its mean and variance by computing the sample mean and variance of ${g}_{k,m}^t$ as
    \begin{align}
        \mu_k^t = \frac{1}{M}\sum_{m=1}^M g_{k,m}^t
        \label{eq:prior1}
    \end{align} 
    and
    \begin{align}
        \left(\nu_k^t\right)^2 = \frac{1}{M}\sum_{m=1}^M \left(g_{k,m}^t\right)^2 -\left(\mu_k^t\right)^2.
        \label{eq:prior2}
    \end{align} 
    Although this prior distribution is simple,  it can still capture distinct statistical information about the prior distribution of each mobile device with two scalars $\mu_k^t$ and $\nu_k^t$ for $k\in [K]$. This simplification significantly reduces both the communication cost in delivering the information about the prior and the complexity in computing the aggregation function. Specifically, each mobile device only needs to send two scalars $\mu_k^t \in \mathbb{R}$ and $\nu_k^t \in \mathbb{R}^{+} $ in addition to its one-bit gradient to the server. Therefore, it considerably diminishes the communication cost compared to sending a large dimensional mean vector ${\bm \mu}_k^t\in \mathbb{R}^M$ and covariance matrix ${\bm \Sigma}_k^t \in \mathbb{R}^{M\times M}$ required for BFL. The mobile devices perform zero-mean normalization before one-bit quantization, i.e.,	
\begin{align}
	{\bar g}_{k,m}^t ={g}_{k,m}^t - \mu_k^t.
\end{align}    

{\bf Quantization:} After normalization, mobile device $k\in[K]$ compresses ${\bf \bar g}_{k}^t$ to a binary vector using the one-bit quantizer, ${\bf \tilde g}_{k}^t={\sf sign}\left(	{\bf \bar g}_{k}^t \right)$. The mobile devices also quantize $\mu_k^t$ and $\nu_k^t$ using a $B$-bit scalar quantizer. Let $\mathcal{Q}\!=\!\{q_1,q_2, \ldots, q_{2^B}\}$ and $\mathcal{B}\!=\!\{b_0,b_1,\ldots, b_{2^B}\}$ be sets of quantized outputs and bin boundaries, respectively. Then, the quantization function ${\sf Q}_B: \mathbb{R}\rightarrow \mathcal{Q}$ maps an input to a discrete-valued output in $\mathcal{Q}$ as 
\begin{align}
	{Q}_{\sf B}(x) = \sum_{\ell=1}^{2^B} q_\ell {\bf 1}_{\left\{ b_{\ell} \leq x \leq b_{\ell+1}\right\}}. \label{eq:Qbit}
\end{align}
 Let ${\hat \mu}_k^t=Q_{\sf B}({ \mu}_k^t)$ and ${\hat \nu}_k^t=Q_{\sf B}({ \nu}_k^t) $ be the quantizer output of ${ \mu}_k^t$ and ${ \nu}_k^t$, respectively. 
 
 {\bf Uplink transmission:} The uplink transmission packet structure is illustrated in Fig. \ref{fig1}. As can be seen, mobile device $k\in[K]$ sends ${\sf sign}\left(	{\bf \bar g}_{k}^t \right)$ with  $\mu_k^t$ and $\nu_k^t$ to the server at communication round $t$. When sending ${\hat \mu}_k^t$ and ${\hat \nu}_k^t$, device $k$ encodes 2$B$ information bits about ${\hat \mu}_k^t$ and ${\hat \nu}_k^t$ using a powerful channel code (e.g., polar codes) so that the server perfectly decodes them. For instance, when a code rate is fixed to $r<1$,  the total of $\frac{2B}{r}$ binary coded symbols are additionally transmitted in conjunction with $M$ binary symbols for sending ${\sf sign}\left(	{\bf \bar g}_{k}^t \right)$. Since model size $M \geq 10^{5}$ is much larger than $\frac{2B}{r} <50$, this additional communication overhead is negligible.

    {\bf Scalable Bayesian aggregation:} When receiving $y_{k,m}^t$ at the server, our strategy is to perform element-wise MMSE estimation using both ${\hat \mu}_k^t$ and ${\hat \nu}_k^t$. Our simplified prior distribution strategy allows obtaining a closed-form expression for a non-linear Bayesian aggregation function.



            \begin{figure*}[t]
    	\centering
        \includegraphics[width=1\linewidth]{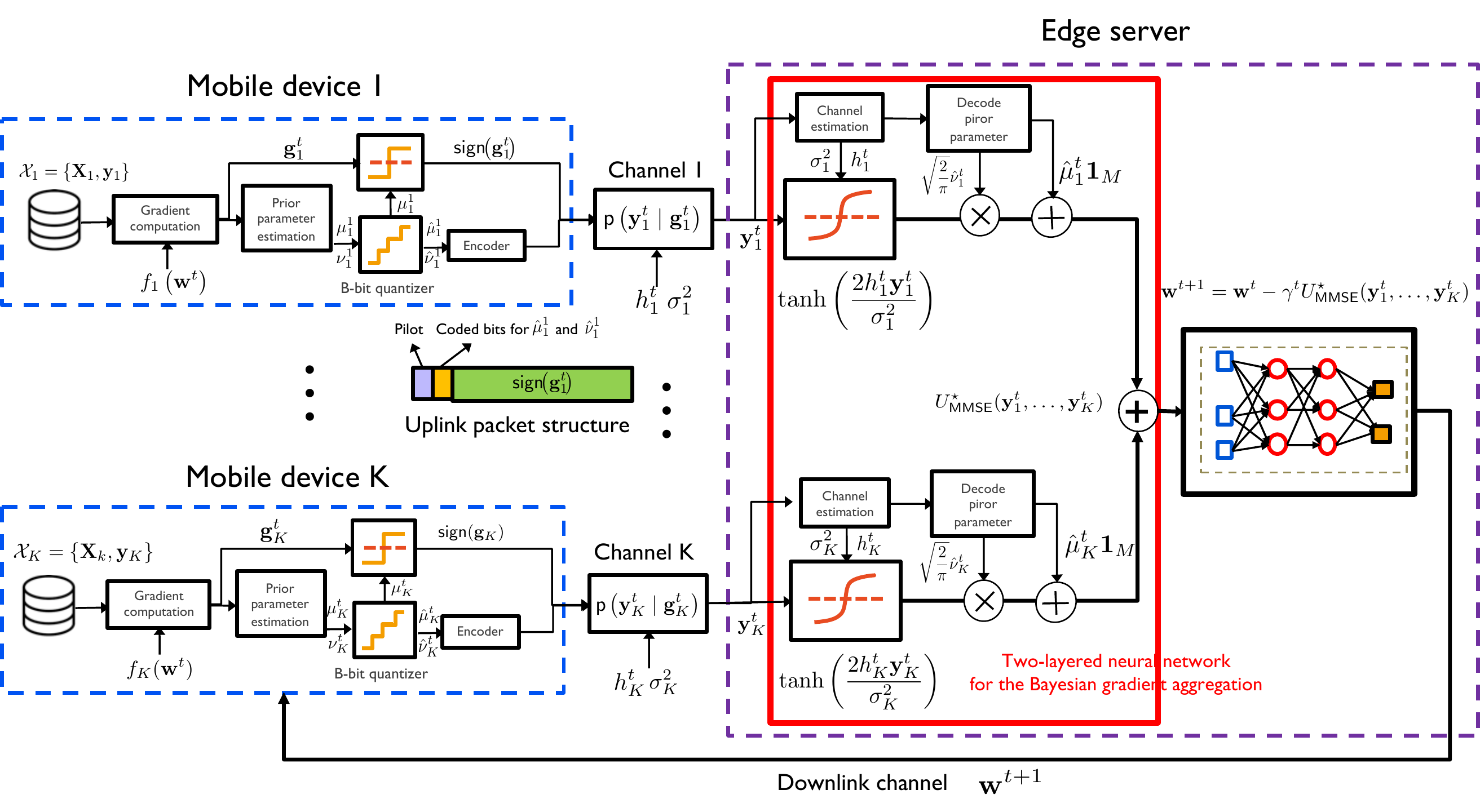}\vspace{-0.1cm}
      \caption{An illustration of the proposed scalable Bayesian gradient aggregator, which can be implemented as a simple neural network.}
      \label{fig1}\vspace{-0.5cm}
    \end{figure*}    
 
    {\bf Proposition 1:} Let ${g}_{k,m}^t$ be iid Gaussian random variables with ${\hat \mu}_k^t$ and ${\hat \nu}_k^t$ for $k\in [K]$ and $m\in [M]$. Also, let $U({\bf y}_1^t,\ldots, {\bf y}_K^t):\mathbb{R}^{MK}\rightarrow\mathbb{R}^M$ be a gradient aggregation function. Then, the aggregation function that minimizes $\mathbb{E}\left[\left\|{\bf g}_{\Sigma}^t-U({\bf y}_1^t,\ldots, {\bf y}_K^t)\right\|_2^2\right]$
  is
    \begin{align}
        U_{\sf MMSE}^{\star}({\bf y}_1^t,\ldots, {\bf y}_K^t)
        &= \sum_{k=1}^K\left\{ {\hat \mu}_k^t{\bf 1}_M+ \sqrt{\frac{2}{\pi}}{\hat \nu}_k^t\tanh\left(\frac{2h_k^t {\bf y}^t_{k}}{\sigma_k^2}\right) \right\}, \label{eq:mmseAgg}
    \end{align}
            	where ${\bf 1}_M$ is the all-ones vector with dimension $M$. 

    \proof
    See Appendix \ref{proof1}. 
    \endproof

     The proposed aggregation method is implemented with a two-stage operation. In the first stage, it performs the non-linear mapping from ${\bf y}_k^t$ to ${\hat \nu}_k^t\sqrt{\frac{2}{\pi}} \tanh\left(\frac{2h_k^t {\bf y}^t_{k}}{\sigma_k^2}\right)$ using the prior distribution parameter $({\hat \mu}_{k}^t,{\hat \nu_k}^t)$ and the communication channel parameters $h_k^t$ and $\sigma_k^2$. In the second stage, it sums ${\hat \nu}_k^t\sqrt{\frac{2}{\pi}} \tanh\left(\frac{2h_k^t {\bf y}^t_{k}}{\sigma_k^2}\right)$ with the means to obtain the estimate of the local gradient ${\bf \bar g}_{k}^t$.       To implement SBFL, the server requires to estimate the communication channel parameters $h_k^t$ and $\sigma_k^2$. The channel parameter $h_k^t$ can be estimated using conventional pilot signaling, and the estimation accuracy is linearly proportional to the length of the pilot signals. Therefore, it is possible to accurately estimate $h_k^t$ at the expense of uplink spectral efficiency. This algorithm is summarized in Algorithm \ref{alg:SBFL_noDC}. One interesting observation is that, as illustrated in Fig. \ref{fig1}, our aggregation function can be implemented with a two-layered neural network.

When ${\bf \bar g}_{k}^t$ and ${\bf \bar g}_{\ell}^t$ for $k\neq \ell\in [K]$ are correlated Gaussian, the derived aggregation function in \eqref{eq:mmseAgg} is suboptimal because under this correlated prior assumption $\sum_{k=1}^K\mathbb{E}\left[	{\bf  \bar g}_{k}^t | {\bf y}_1^t,\ldots,  {\bf y}_K^t\right]\neq \sum_{k=1}^K\mathbb{E}\left[	{\bf \bar g}_{k}^t | {\bf y}_k^t\right]$ as explained in \eqref{eq:mmse}. Nevertheless, under the correlated prior distribution of local gradients, SBFL is still a promising algorithm due to its computational scalability in both the number of mobile devices and the model size at the cost of a MSE performance loss.

    \begin{algorithm}[tb]
       \caption{SBFL}
       \label{alg:SBFL_noDC}
        \begin{algorithmic}[1]
           \REQUIRE Learning rate $\gamma^0$, momentum $\delta^t$, $K$ mobile devices
           \STATE {\bfseries Initialize:}  ${\bf w}^0$; ${\bf m}^{0}={\bf 0}$
           \FOR{$t = 0,\dots,T-1$}
           \STATE ~~{\bfseries on each device} $k$
           \STATE ~~~~~~{\bfseries receive }${\bf w}^t$
           \STATE ~~~~~~{\bfseries compute }${\bf g}_k^{t} = \nabla f_k({\bf w}^{t})$, $\mu_k^t$, and $\nu_k^t$
           \STATE ~~~~~~{\bfseries send }${\sf sign}({\bf g}_k^{t})$, ${\hat \mu}_k^t=Q_{\sf B}\left({ \mu}_k^t\right)$, and ${\hat \nu}_k^t=Q_{\sf B}\left({ \nu}_k^t\right)$
           \STATE ~~{\bfseries on server}
           \STATE ~~~~~~{\bfseries receive} ${\bf y}_k^t$, ${\hat \mu}_k^t$, and ${\hat \nu}_k^t$ for $k\in [K]$
           \STATE ~~~~~~${\bf m}^{t+1}=\delta^t {\bf m}^{t}+\sum_{k=1}^K\left\{ {\hat \mu}_k^t{\bf 1}_M+ \sqrt{\frac{2}{\pi}}{\hat \nu}_k^t\tanh\left(\frac{2h_k^t {\bf y}^t_{k}}{\sigma_k^2}\right) \right\}$
           \STATE ~~~~~~${\bf w}^{t+1}= {\bf w}^t -\gamma^t {\bf m}^{t+1}$
           \STATE ~~~~~~{\bfseries send }${\bf w}^{t+1}$
           \ENDFOR
        \end{algorithmic}
    \end{algorithm}\vspace{-0.5cm}

    \subsection{Special Case and Generalizations}
    To better understand the scalable Bayesian aggregator, it is instructive to consider some special cases and possible generalizations.
    \vspace{0.2cm}

      \subsubsection{Bussgang-based Linear MMSE Aggregation}
    
    Another strategy is to use the Bussgang-based linear MMSE (BLMMSE) aggregation function.  Using Bussgang's theorem \cite{Bussgang}, the one-bit quantization output can be represented as a linear combination of the quantization input and quantization noise as  
\begin{align}
	{\bf \tilde g}_k^t={\sf sign}({\bf \bar g}_k^t) =  {\bf B}_k^t{\bf \bar g}_k^t +{\bf q}_k^t,\label{eq:linear}
\end{align}
where ${\bf B}_k^t\in \mathbb{R}^{M\times M}$ is a linear quantization operator called the Bussgang matrix and ${\bf q}_k^t\in \mathbb{R}^{M\times 1}$ is the quantization noise, which is uncorrelated with ${\bf \bar g}_k^t $. Under the premise that ${\bf \bar g}_k^t$ is correlated Gaussian random vector, i.e., ${\bf \bar g}_k^t\sim \mathcal{N}\left({\bf 0},{\bm {\Sigma}}_k^t\right)$, the Bussgang matrix that minimizes the quantization error under the LMMSE criterion is given by \cite{Bussgang}:
\begin{align}
	{\bf B}_k^t=\sqrt{\frac{2}{\pi}} {\rm diag}\left({\bm {\Sigma}}_k^t\right)^{-\frac{1}{2}}.\label{eq:Buss}
\end{align}
Also, let ${\bm {\tilde \Sigma}}_k^t=\mathbb{E}\left[{\bf \tilde g}_k^t({\bf \tilde g}_k^t)^{\top}\right]$ be the autocorrelation matrix of the one-bit quantization output. By the arcsin law \cite{Bussgang2}, the autocorrelation matrix is given by
\begin{align}
	{\bm {\tilde \Sigma}}_k^t=\frac{2}{\pi}\left[{\rm arcsin}\left( {\rm diag}({\bm {\Sigma}}_k^t)^{-\frac{1}{2}}{\bm {\Sigma}}_k^t  {\rm diag}({\bm {\Sigma}}_k^t) ^{-\frac{1}{2}}\right)\right]. \label{eq:arsin}
\end{align}
Using ${\bm {\Sigma}}_k^t$, ${\bf B}_k^t$ in \eqref{eq:Buss}, and ${\bm {\tilde \Sigma}}_k^t$ in \eqref{eq:arsin}, the covariance matrix of ${\bf q}_k^t$ is computed as
\begin{align}
	{\bf C}_{{\bf q}_k^t{\bf q}_k^t}={\bm {\tilde \Sigma}}_k^t -{\bf B}_k^t{\bm {\Sigma}}_k^t({\bf B}_k^t)^{\top}.
\end{align}
Plugging the linearized quantization model of \eqref{eq:linear} into \eqref{eq:channel}, the received signal of the $k$th subchannel at the server can be rewritten as
\begin{align}
	{\bf y}_k^t= h_k^t\left({\bf B}_k^t{\bf \bar g}_k^t +{\bf q}_k^t\right)+{\bf n}_k^t.
\end{align}
Therefore, the resultant BLMMSE estimator of ${\bf \bar g}_{k}^t$ given ${\bf y}_k^t$ is
\begin{align}
	\mathbb{\hat E}\left[	{\bf \bar g}_{k}^t | {\bf y}_k^t\right]&=h_k^t {\bm {\Sigma}}_k^t ({\bf B}_k^t)^{\top}\left((h_k^t)^2{\bf B}_k^t {\bm {\Sigma}}_k^t ({\bf B}_k^t)^{\top}+(h_k^t)^2{\bf C}_{{\bf q}_k^t{\bf q}_k^t} +\sigma_k^2{\bf I}\right)^{-1}{\bf y}_k^t \nonumber\\
	&=h_k^t {\bm {\Sigma}}_k^t ({\bf B}_k^t)^{\top}\left((h_k^t)^2{\bm {\tilde \Sigma}}_k^t  +\sigma_k^2{\bf I}\right)^{-1}{\bf y}_k^t \nonumber\\
	&=h_k^t {\bm {\Sigma}}_k^t \sqrt{\frac{2}{\pi}} {\rm diag}\left({\bm {\Sigma}}_k^t\right)^{-\frac{1}{2}}\left((h_k^t)^2\frac{2}{\pi}\left[{\rm arcsin}\left( {\rm diag}({\bm {\Sigma}}_k^t)^{-\frac{1}{2}}{\bm {\Sigma}}_k^t  {\rm diag}({\bm {\Sigma}}_k^t) ^{-\frac{1}{2}}\right)\right] +\sigma_k^2{\bf I}\right)^{-1}{\bf y}_k^t. \label{eq:BLMMSE}
\end{align}
To compute $\mathbb{\hat E}\left[	{\bf \bar g}_{k}^t | {\bf y}_k^t\right]$, the server requires knowledge of ${\bm {\Sigma}}_k^t\in \mathbb{R}^{M\times M}$ for $k\in [K]$, which increases uplink communication cost significantly.  The following proposition provides a scalable expression of the BLMMSE aggregation function under the iid Gaussian prior assumption ${\bf \bar g}_k^t\sim \mathcal{N}\left({\bf 0},{\nu}_k^t{\bf I}_M\right)$.

    
        {\bf Proposition 2:} Let ${g}_{k,m}^t$ be iid Gaussian random variables with ${\hat \mu}_k^t$ and ${\hat \nu_k}^t$ for $k\in [K]$ and $m\in [M]$. Also, let $U({\bf y}_1^t,\ldots, {\bf y}_K^t):\mathbb{R}^{MK}\rightarrow\mathbb{R}^M$ be a {\it linear} aggregation function. Then, the optimal {\it linear} function that minimizes $ \mathbb{E}\left[\left\|{\bf g}_{\Sigma}^t-U_({\bf y}_1^t,\ldots, {\bf y}_K^t)\right\|_2^2\right]$ is given by
    \begin{align}
        U^{\star}_{\sf BLMMSE}({\bf y}_1^t,\ldots, {\bf y}_K^t)= \sum_{k=1}^K\left\{ {\hat \mu}_k^t{\bf 1}_M+    \frac{\sqrt{\frac{2}{\pi}}h_k^t{\hat \nu}_k^t}{ \frac{2}{\pi}(h_k^t)^2  +\sigma_k^2 }{\bf y}_k^t\right\}.\label{eq:blmmse2}
    \end{align}

 \proof
Plugging ${\bf B}_k^t=\sqrt{\frac{2}{\pi}} \frac{1}{\nu_k^t}{\bf I}$, ${\bm {\tilde \Sigma}}_k^t={\bf I}$, 	and ${\bf C}_{{\bf q}_k^t{\bf q}_k^t}=\left(1-\frac{2}{\pi}\right){\bf I}$ into \eqref{eq:BLMMSE}, we obtain \eqref{eq:blmmse2}, which completes the proof.      \endproof      	

This proposition shows that the BLMMSE aggregation function is also scalable to the number of mobile devices and the model size because it requires element-wise operation. However, this aggregation function is worse than the MMSE aggregation function in terms of the MSE performance because it minimizes the MSE under a linear map constraint.

    \subsubsection{Laplacian prior}
    When a local gradient vector is sparse, the Laplacian prior distribution can capture the gradient's sparsity structure better than the Gaussian prior  \cite{glorot2010understanding,shi2019understanding}. The following proposition shows how to change the Bayesian aggregation function for the Laplacian prior. 
    
    \vspace{0.2cm}
    {\bf Proposition 3:} Let ${\bar g}_{k,m}^t$ be  iid Laplacian random variables with zero-mean and scale parameter $\lambda_k^t$:
    \begin{align}
        P\left( {\bf \bar g}_{k}^t\right)=\prod_{m=1}^M\frac{1}{2\lambda_k^t}\exp\left(-\frac{|{\bar g}_{k,m}^t|}{\lambda_k^t}\right). 
        \label{eq:Laplacian}
    	\end{align}
Then, the MMSE aggregation function  is
    \begin{align}
     U_{\sf MMSE,Lap}^{\star}({\bf y}_1^t,\ldots, {\bf y}_K^t)= \sum_{k=1}^K\left\{ \mu_k^t{\bf 1}_M +\lambda_k^t  \tanh\left(\frac{2h_k^t  {\bf y}^t_{k}}{\sigma_k^2}\right)\right\}.
    	\label{eq:mmse_laplace} 
    \end{align}
    \begin{proof}
        The proof is direct from Appendix \ref{proof1} by changing the prior distribution from Gaussian to Laplacian.
    \end{proof}
    This gradient aggregation function is almost identical to the aggregation function assuming the Gaussian prior distribution in \eqref{eq:mmseAgg} except for a scaling parameter.  
    \vspace{0.2cm}

    \subsubsection{High SNR regime}
    The following corollary shows how to simplify the estimator in the high SNR regime. 
    
    \begin{cor}\label{cor1}
    	When $\sigma_k^2\rightarrow 0$, the aggregation functions in \eqref{eq:mmseAgg} and \eqref{eq:blmmse2} simplify to
    \begin{align}
         \lim_{\sigma_k^2\rightarrow 0}   U_{\sf MMSE}^{\star}({\bf y}_1^t,\ldots, {\bf y}_K^t)=  \sum_{k=1}^K\left\{ {\hat \mu}_k^t{\bf 1}_M+ {\hat \nu_k}^t \sqrt{\frac{2}{\pi}}{\sf sign}\left(\frac{{\bf y}_k^t}{h_k^t}\right)\right\}. 
        \label{eq:mmseAgg2} 
        \end{align}
        and 
            \begin{align}
  \lim_{\sigma_k^2\rightarrow 0}   U_{\sf BLMMSE}^{\star}({\bf y}_1^t,\ldots, {\bf y}_K^t)=  \sum_{k=1}^K\left\{ {\hat \mu}_k^t{\bf 1}_M+ {\hat \nu_k}^t \sqrt{\frac{2}{\pi}} \frac{{\bf y}_k^t }{h_k^t} \right\}. 
        \label{eq:blmmseAgg3} 
        \end{align}
    \end{cor}
    
    \begin{proof}
        The proof is direct from the definition of the hyperbolic tangent function.
    \end{proof}
    	   	
    This corollary shows that the proposed aggregation function becomes a weighted sum of the local gradient signs.  Therefore, for the noise-free channel case, one can exploit the standard deviation of the local gradient ${\hat \nu}_k^t$ as the weights of the heterogeneous local gradients because ${\sf sign}\left(h_k^t{\bf y}_k^t\right)={\sf sign}\left(\frac{{\bf y}_k^t}{h_k^t}\right)={\sf sign}\left({\bf \bar g}_k^t\right)$. Furthermore, if ${\hat \mu}_k^t=0$ and ${\hat \nu}_k^t$ are equal across mobile devices, SBFL reduces to the signSGD algorithm in \cite{bernstein2018signsgd}. Consequently, SBFL generalizes the signSGD algorithm by incorporating data heterogeneity of mobile devices from a Bayesian viewpoint. In addition, when $h_k^t=1$ and $\sigma_k^2=0$, the BLMMSE aggregation function becomes identical to the MMSE aggregation function because ${\bf y}_k^t={\sf sign}({\bf \bar g}_k^t)$.

    \vspace{0.2cm}

    \vspace{0.2cm}
    \subsubsection{Extension with downlink compression}\label{algorithm_DC}
    We have assumed that the server can send the real-valued model parameters ${\bf w}^t$ to the mobile devices perfectly via the downlink communication channels. When the model size is extremely large, sending the real-valued ${\bf w}^t$ increases the downlink communication cost. Thereby, the server requires compressing the real-valued model parameter to reduce the downlink communication cost. To accomplish this, we can modify the proposed SBFL. Instead of performing the model update at the server, we can compress the aggregated local gradients to a binary vector, and send the compressed sign information to mobile devices via downlink channels. Using the received binary information, each mobile device independently updates the model parameters using the same learning rate. This algorithm modification can significantly reduce the downlink communication cost because it sends a binary vector using the downlink channel per communication round.  We summarize this in Algorithm \ref{alg:SBFL_DC}. 
 
  
    
    \begin{algorithm}[t]
      \caption{SBFL with Downlink Compression}
      \label{alg:SBFL_DC}
    \begin{algorithmic}[1]
      \REQUIRE Learning rate $\gamma$, momentum $\delta$, $K$ mobile devices
      \STATE {\bfseries Initialize:} ${\bf w}_k^{-1} = {\bf w}^{-1}$, ${\bf m}^{-1}_k={\bf 0}$, and $ {\bf y}_k^{-1}={\bf 0}~\forall k \in [K]$;
      \FOR{$t = 0,\dots,T-1$}
      \STATE ~~{\bfseries on each device} $k$
      \STATE ~~~~~~${\bf m}^{t}_k = \delta {\bf m}^{t-1}_k + {\sf sign}\left( \sum_{k=1}^K\left\{ {\hat \mu}_k^t{\bf 1}_M+ \sqrt{\frac{2}{\pi}} {\hat \nu}_k^t\tanh\left(\frac{2h_k^t {\bf y}^{t-1}_{k}}{\sigma_k^2}\right) \right\}\right)$
      \STATE ~~~~~~${\bf w}^{t}_k = {\bf w}^{t-1}_k - \gamma {\bf m}^{t}_k$
      \STATE ~~~~~~{\bfseries compute }${\bf g}_k^{t} = \nabla f_k\left({\bf w}_k^{t}\right)$, $\mu_k^{t}$, and $\nu_k^{t}$
      \STATE ~~~~~~{\bfseries send }${\sf sign}\left({\bf g}_k^{t}\right)$, ${\hat \mu}_k^{t}$, and ${\hat \nu}_k^{t}$
      \STATE ~~{\bfseries on server}
      \STATE ~~~~~~{\bfseries receive} ${\bf y}_k^t$, ${\hat \mu}_k^{t}$, and ${\hat \nu}_k^{t}$ for $k\in [K]$
      \STATE ~~~~~~{\bfseries send }${\sf sign}\left( \sum_{k=1}^K\left\{ {\hat \mu}_k^t{\bf 1}_M+ \sqrt{\frac{2}{\pi}} {\hat \nu}_k^t\tanh\left(\frac{2h_k^t {\bf y}^t_{k}}{\sigma_k^2}\right) \right\}\right)$
      \ENDFOR
    \end{algorithmic}
    \end{algorithm}

    \section{Performance Analysis}
    This section provides a convergence analysis of SBFL for a class of convex loss functions. The critical step in the analysis is to derive that MSE values according to different gradient aggregation functions. Then, leveraging the derived MSE values, we show that the gradient descent algorithm employing the proposed aggregation method guarantees convergence to a stationary point for a class of non-convex loss functions. 
    
    \subsection{MSE Bounds}
    We characterize the optimal MSE under the premise that all local gradients' prior distributions have zero-means, i.e., $\mu_k^t=0$ for $k\in [K]$ and $t\in [T]$.
    
    \begin{thm}\label{thm1}  Let ${\bar g}_{k,m}^t$ be iid Gaussian with zero-mean and $\nu_k^t$ for $k\in [K]$ and $m\in [M]$. Then, the minimum MSE, $\eta_{\sf mse}^t \triangleq  \mathbb{E}\left[\left\|{\bf \bar g}_{\Sigma}^t- U_{\sf MMSE}^{\star}({\bf y}_1^t,\ldots, {\bf y}_K^t)\right\|_2^2\right]$, is 
    \begin{align}
    	 \eta_{\sf mse}^t=M\sum_{k=1}^K (\nu_k^t)^2\left[1-\frac{2}{\pi} \int_{-\infty}^{\infty}  \tanh\left(\frac{2h_k^ty_k^t}{\sigma_k^2}\right)^2P({y}_{k,m}^t) {\rm d}y_{k,m}^t\right],\label{eq:mse}
    \end{align}
    where 
    \begin{align}
    	P({y}_{k,m}^t)=\frac{\exp\left(-\frac{|y_{k,m}^t-h_k^t|^2}{2\sigma_k^2}\right)+\exp\left(-\frac{|y_{k,m}^t+h_k^t|^2}{2\sigma_k^2}\right)}{2\sqrt{2\pi}\sigma_k}.\label{eq:py}
    \end{align}
    \end{thm}
    \begin{proof}
     	See Appendix \ref{proof2}.
    \end{proof}
Theorem \ref{thm1} shows that the minimum MSE value is proportional to $({\nu}_{k,}^t)^2$ and $M$, and it decreases by a factor of $\left[1-\frac{2}{\pi} \int_{-\infty}^{\infty}  \tanh\left(\frac{2h_k^ty_k^t}{\sigma_k^2}\right)^2P({y}_{k,m}^t) {\rm d}y_{k,m}^t\right]$ because  $\int_{-\infty}^{\infty}  \tanh\left(\frac{2h_k^ty_k^t}{\sigma_k^2}\right)^2P({y}_{k,m}^t) {\rm d}y_{k,m}^t<1$.  The following corollary establishes a closed-form expression for an upper bound on the minimum MSE derived in Theorem \ref{thm1}.

         \begin{cor} \label{cor2} When using $\lim_{\sigma_k^2\rightarrow 0}U_{\sf MMSE}^{\star}({\bf y}_1^t,\ldots, {\bf y}_K^t)=\sum_{k=1}^K{\hat \nu_k}^t \sqrt{\frac{2}{\pi}}{\sf sign}\left(\frac{{\bf y}_k^t}{h_k^t}\right)$ in \eqref{eq:mmseAgg2}, the minimum MSE becomes
\begin{align}
\mathbb{E}\left[\left\|{\bf \bar g}_{\Sigma}^t- \lim_{\sigma_k^2\rightarrow 0}U_{\sf MMSE}^{\star}({\bf y}_1^t,\ldots, {\bf y}_K^t)\right\|_2^2\right]= M\left(1-\frac{2}{\pi} \right)\left(\sum_{k=1}^K (\nu_k^t)^2\right).
\end{align}
         \end{cor}

\begin{proof}
The proof is direct by replacing $U_{\sf MMSE}^{\star}({\bf y}_1^t,\ldots, {\bf y}_K^t)$ to $ \lim_{\sigma_k^2\rightarrow 0}U_{\sf MMSE}^{\star}({\bf y}_1^t,\ldots, {\bf y}_K^t)$ in Appendix \ref{proof2}. 
\end{proof}

We also derive the minimum MSE when using the BLMMSE aggregation functions.

         \begin{cor} \label{cor3} When using $U_{\sf BLMMSE}^{\star}({\bf y}_1^t,\ldots, {\bf y}_K^t)$ and $\lim_{\sigma_k^2\rightarrow 0}U_{\sf BLMMSE}^{\star}({\bf y}_1^t,\ldots, {\bf y}_K^t)$, the minimum MSE values are given by \begin{align} 
\mathbb{E}\left[\left\|{\bf \bar g}_{\Sigma}^t- U_{\sf BLMMSE}^{\star}({\bf y}_1^t,\ldots, {\bf y}_K^t)\right\|_2^2\right]=M\sum_{k=1}^K (\nu_k^t)^2\left[1- \frac{ \frac{2}{\pi} (h_k^t)^2}{  \frac{2}{\pi} (h_k^t)^2+\sigma_k^2}\right]
\end{align}
and 
\begin{align} 
\mathbb{E}\left[\left\|{\bf \bar g}_{\Sigma}^t- \lim_{\sigma_k^2\rightarrow 0}U_{\sf BLMMSE}^{\star}({\bf y}_1^t,\ldots, {\bf y}_K^t)\right\|_2^2\right]=M\left(1-\frac{2}{\pi} \right)\left(\sum_{k=1}^K (\nu_k^t)^2\right).
\end{align}
         \end{cor}
             \begin{proof}
     	See Appendix \ref{proof3}.
    \end{proof}
          Notice that when using $\lim_{\sigma_k^2\rightarrow 0}U_{\sf BLMMSE}^{\star}({\bf y}_1^t,\ldots, {\bf y}_K^t)$, the MSE becomes identical to that when applying $\lim_{\sigma_k^2\rightarrow 0} U_{\sf MMSE}^{\star}({\bf y}_1^t,\ldots, {\bf y}_K^t)$ as the aggregation function. This implies that, the BLMMSE aggregation function is optimal when the SNRs of all communication links are infinity.

    \subsection{Convergence Analysis} \label{convergence analysis}
 To facilitate the convergence analysis of SBFL, we commence by introducing assumptions that describe some properties of a non-convex loss function. 
    
    {\bf Assumption 1:} For any parameter ${\bf w}$, the loss function is bounded below by some value $f({\bf w}^{\star})$, i.e., $f({\bf w})\geq f({\bf w}^{\star})$ for all ${\bf w}$.
    
        {\bf Assumption 2:} The loss function is $L$-Lipschitz smooth, i.e., for any ${\bf x}$ and ${\bf y}$, 
            \begin{align}
        f({\bf y}) \leq f({\bf x}) + \nabla f({\bf x})^{\top}({\bf y}-{\bf x}) + \frac{L}{2}\|{\bf y}-{\bf x}\|_2^2.
    \end{align}

 Using the above assumptions and Theorem \ref{thm1},  the following theorem shows the convergence rate of SBFL when optimizing a class of non-convex and smooth loss functions.
    
        \begin{thm}\label{thm2}
        	  Let $f:\mathbb{R}^M\rightarrow\mathbb{R}$ be a $L$-Lipschitz smooth and non-convex loss function. SBFL with learning rate $\gamma^t=\frac{\gamma}{\sqrt{t+1}}$ for $\gamma>0$ and momentum $\delta^t=0$ satisfies
        	  \begin{align}
	\mathbb{E}\left[ \frac{1}{T}\sum_{t=0}^{T-1} \|{\bf g}_{\Sigma}^t\|_2^2\right]  \leq \frac{1}{\sqrt{T}}\left[ \frac{f\left({\bf w}^{0}\right) -  f\left({\bf w}^{\star}\right)  
}{\gamma \left(1-\frac{\gamma L}{2}\right) } +\sigma^2_{\sf mse}\left(1+\ln T\right)\frac{ \frac{\gamma^2 L}{2}}{1-\frac{\gamma L}{2}}  \right],\label{eq:convergence}
\end{align}
where  $\sigma^2_{\sf mse}\triangleq \max_{t\in [T]}\{ \eta_{\sf mse}^t\}$. 
        \end{thm}
    \begin{proof}
     	See Appendix \ref{proof4}.
    \end{proof}
   Theorem \ref{thm2} implies that the expected value of the squared gradient norm decreases as the number of communications round $T$ grows in the order of 
   \begin{align}
   	\mathcal{O}\left(\frac{c+c'\sigma^2_{\sf mse}\ln(T)}{\sqrt{T}} \right),\label{eq:order}
   \end{align}   
   for some positive constants $c$ and $c'$. When $\sigma^2_{\sf mse}=0$, the convergence rate of the full-batch based gradient decent algorithm reduces to $\mathcal{O}\left(\frac{1}{\sqrt{T}} \right)$. Therefore, larger MSE $\sigma^2_{\sf mse}$ makes the convergence speed slower. By using the MMSE aggregation function that provides the minimum $\sigma^2_{\sf mse}$, we can speed up the convergence rate. Since $\sigma^2_{\sf mse}$ is bounded by a constant, the right-hand side term in \eqref{eq:convergence} approaches zero as $T$ goes infinity because $\lim_{T\rightarrow\infty}\frac{\ln(T)}{\sqrt{T}}=0$. Consequently, the expected value of the squared gradient norm goes to zero  as $T\rightarrow\infty$, implying that SBFL converges to a stationary point of the non-convex loss function with the rate in \eqref{eq:order}.


%

    For ease of exposition, in our proof, we focused on a simple gradient descent algorithm with the proposed Bayesian aggregation method. We can readily extend this proof for SGD using a mini-batch size. The convergence behaviors of SGD using mini-batch gradients in conjunction with the proposed aggregation function will be numerically verified through simulations.  


\section{Simulation Results}
In this section, we evaluate the learning performance of the proposed algorithm and compare it with signSGD \cite{bernstein2018signsgd} to illustrate the synergetic gains of using both the local gradient priors and the channel distributions in a heterogeneous network.  We first explain the simulation settings, including a network model, channel models, and data distributions. We  then provide numerical examples that show the learning performance of the SBFL algorithm in two different scenarios.  We begin with a simple federated learning setting in which a linear classifier is optimized with synthetic yet heterogeneous data sets across users, each with distinct communication link qualities. Then, we consider a more complicated setup, in which a CNN is trained using MNIST datasets allocated to users in a heterogeneous manner. 

   \begin{figure*}[t]
    	\centering
        \includegraphics[width=0.9\linewidth]{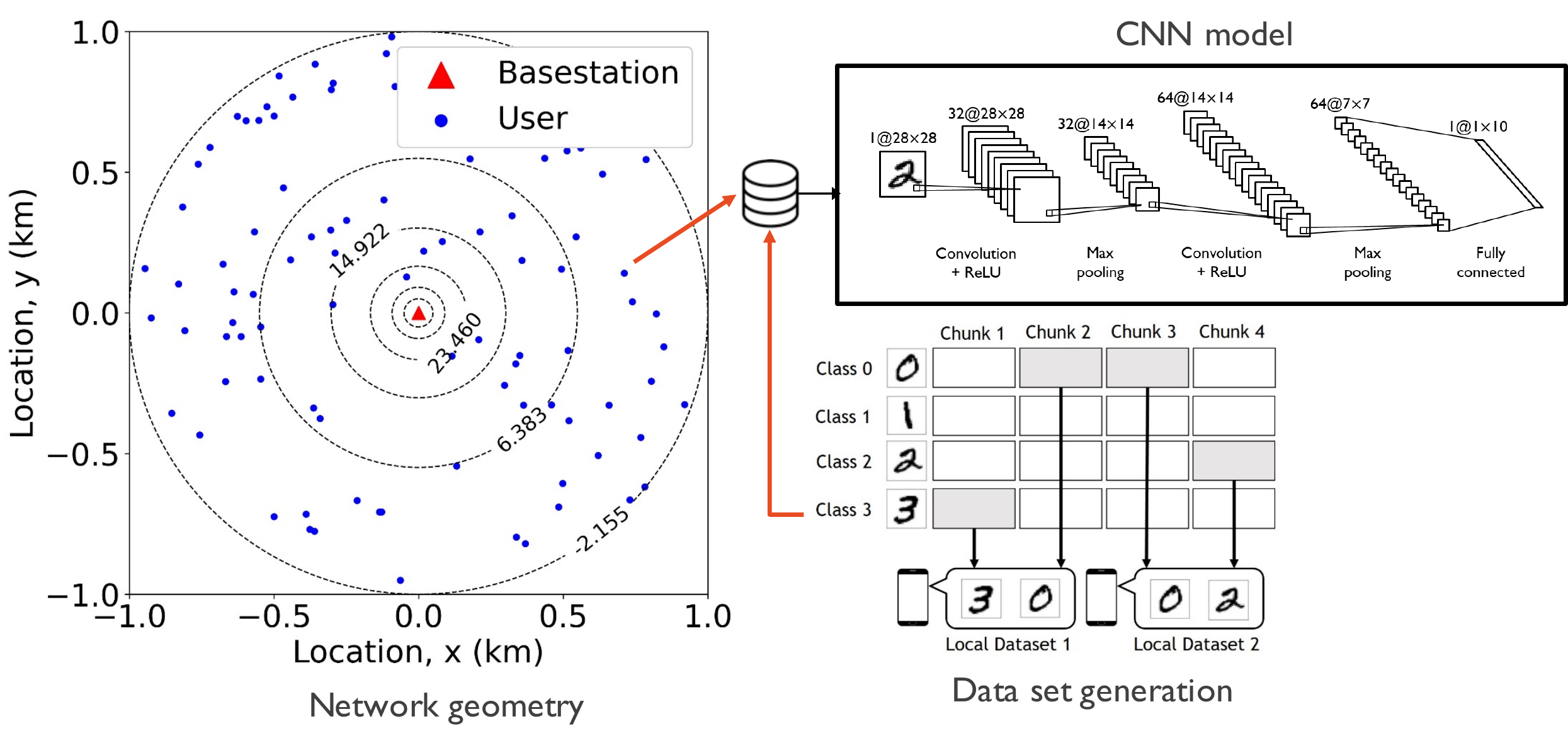}\vspace{-0.3cm}
      \caption{Heterogeneities in the simulation setup. } \label{fig:simulation}\vspace{-0.7cm}
    \end{figure*}
    

    \subsection{Simulation Settings}
    We explain the details of the experimental setup, including learning models, data distributions, and the network model. 
    
    \textbf{Network model:}
    We consider a single-cell scenario in which the locations of users are uniformly distributed in a cell with radius $1$km as illustrated in Fig. \ref{fig:simulation}. The uplink budgets and SNRs are defined according to \cite{mollel2014comparison, sharma2016cell}, in which COST-231 Hata model is used for the path-loss \cite{heath2018foundations}. We also consider an urban environment where a base station and users have heights of $70 $m and $1.5 $m, respectively. Since we focus on the real part of the complex base-band equivalent model, the fading channel coefficients are drawn from iid real Gaussian distribution, i.e., $h_k^t \sim \mathcal{N}(0,1)$. This fading channel is assumed to be constant during a communication round and changes over rounds.

    \textbf{Linear regression with synthetic heterogeneous datasets:}
    We consider a linear regression for a classification task. The local loss function of the linear regression is a simple convex function, i.e., $f_k({\bf w}^t)=\|{\bf X}_k^{\top}{\bf w}^t-{\bf y}_k\|_2^2$ with $[N_k,M,K]=[100,300,20]$. To optimize this linear classifier, we generate  synthetic heterogeneous datasets.  To embrace the data heterogeneity, the data matrix of user $k\in[K]$ is drawn from $\mathcal{N}({\bf 0},a_k{\bf I})$ where the covariance matrix has a scaled identity matrix. Here, the scale parameter is uniformly chosen from (0,5), i.e., $a_k \sim \mathcal{U}(0,5)$. The labels ${\bf y}_k$ are drawn from $\mathcal{N}({\bf 0},{\bf I})$ for $k\in [K]$.  When optimizing the model parameters, we use a step size of  $1/L$, where $L$ is Lipschitz-smoothness of the global loss function.
    
    \textbf{CNN with MNIST datasets:}
    We consider an image classification task using the MNIST dataset by optimizing the CNN model parameters with a cross-entropy loss function. In our simulations, we use a CNN that  consists of two convolutional layers and a linear layer. Each convolutional layer comprises multiple sequential operations, including convolution, ReLu activation, and max-pooling. The kernel sizes of the first and second convolutional layers are set to be $32$ and $64$, respectively. To perform the image classification experiments, we use the MNIST dataset provided by \cite{lecun1998gradient}. To generate heterogeneous datasets across mobile devices, we assume that each class's training samples for the MNIST dataset are divided by the number of users, which is called a chunk of the class. Then, each user uniformly selects two chunks from distinct classes of the MNIST dataset in a non-overlapping manner, as in \cite{mcmahan2017communication} and illustrated in Fig. \ref{fig:simulation}. When applying the SGD algorithm, all users use the identical batch size of $32$ whenever computing gradients.  When optimizing the CNN model parameters, we use the estimated gradient with  momentum because it empirically outperforms the algorithm using gradient only \cite{sattler2019robust}.
    
    To capture the heterogeneity of wireless links, all experiments are repeated $30$ times by changing users' locations and the initial model parameters to observe the robustness of the link heterogeneity. Performance measures are displayed with their standard deviations by the shaded area in the figures.

    \begin{figure*}[t]
    \centering
        \subfigure[]{
        \includegraphics[width=0.222\linewidth]{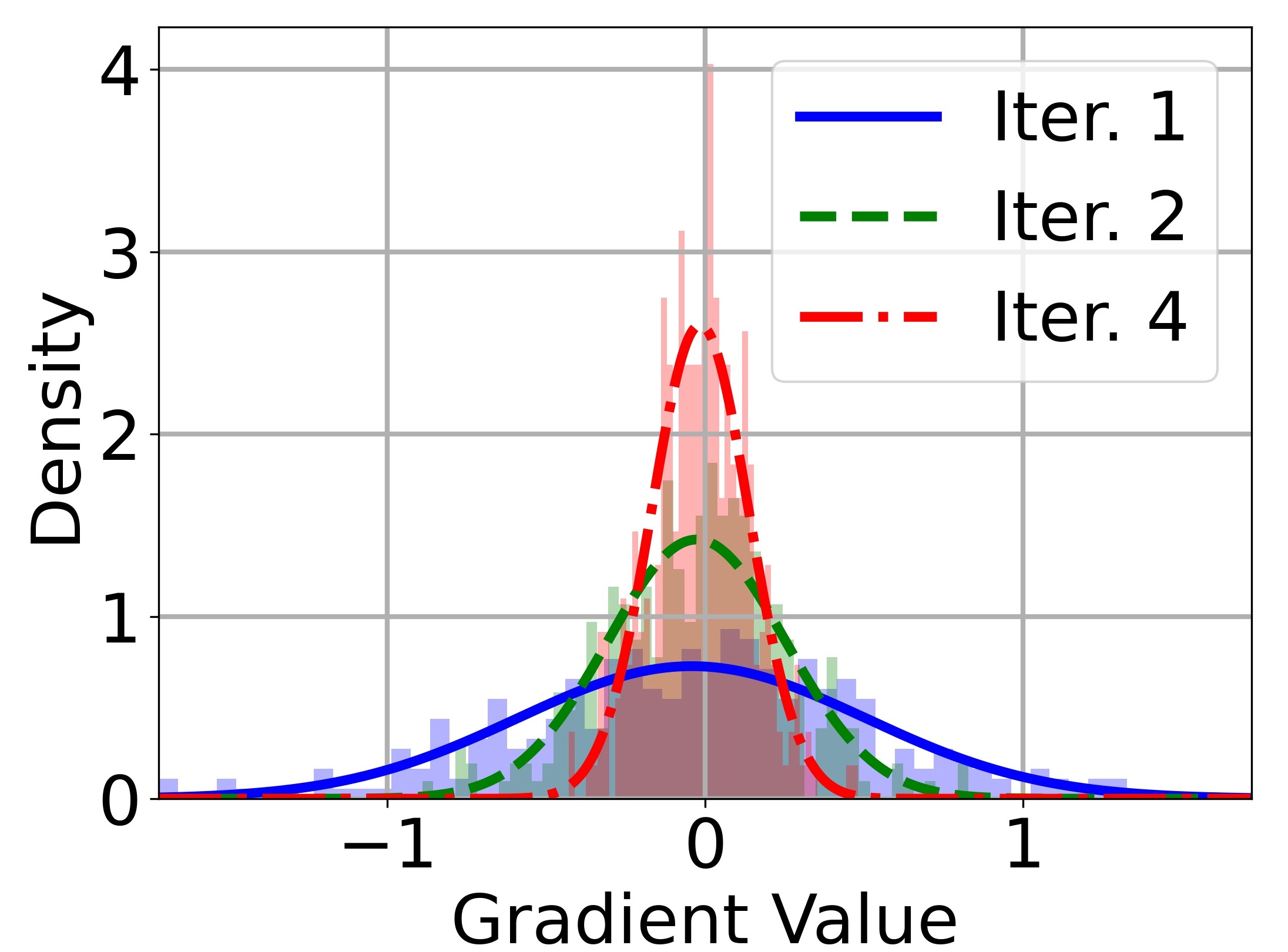}
        \label{fig:hist_syn_iter}
        }
        \subfigure[]{
        \includegraphics[width=0.222\linewidth]{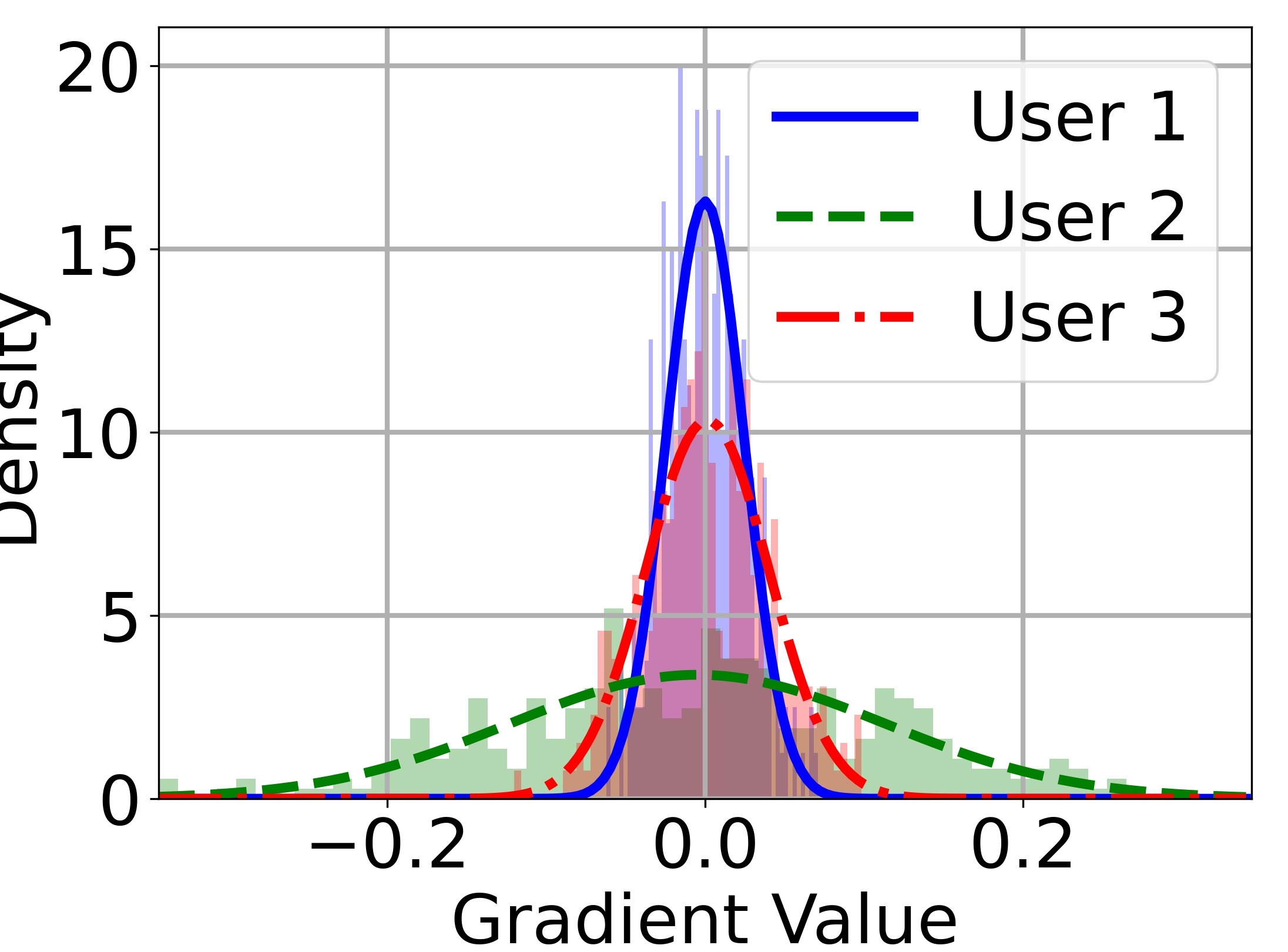}
        \label{fig:hist_syn_user}
        }
        \subfigure[]{
        \includegraphics[width=0.225\linewidth]{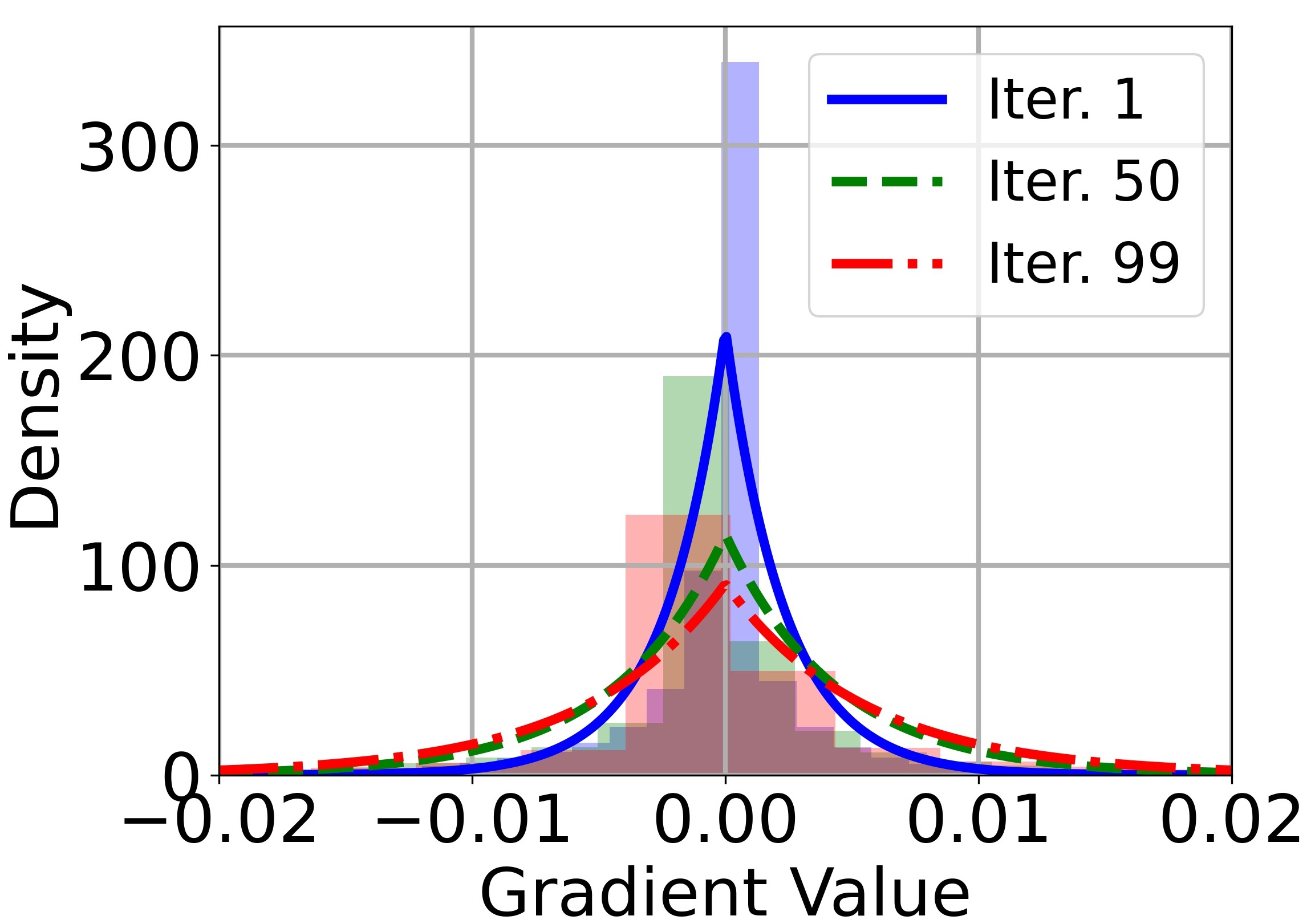}
        \label{fig:hist_mnist_iter}
        }
        \subfigure[]{
        \includegraphics[width=0.225\linewidth]{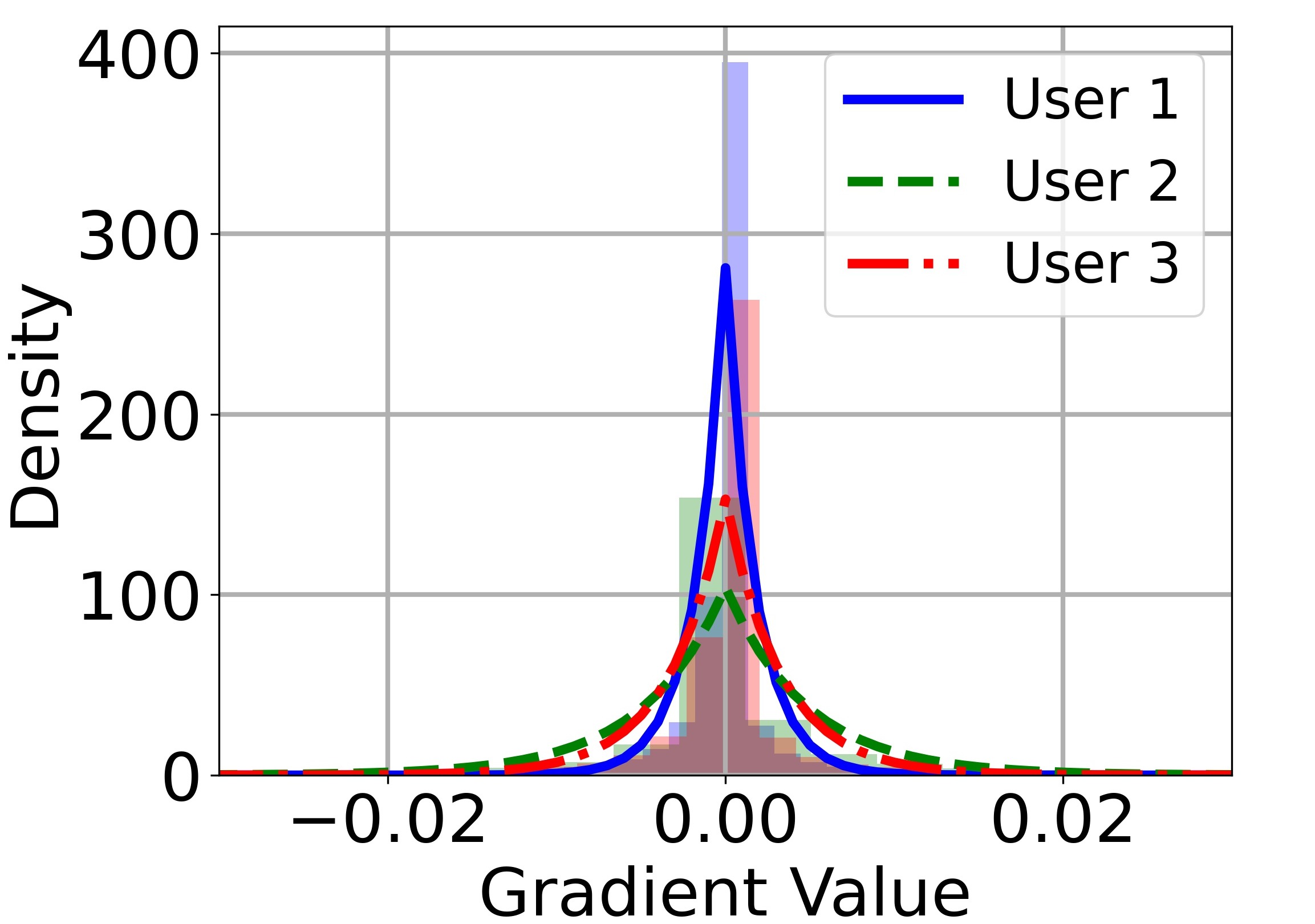}
        \label{fig:hist_mnist_user}
        }    \vspace{-0.2cm}
        \caption{Histograms for the elements of ${\bf g}_k^t$ and the corresponding Gaussian and Laplacian prior approximations with a proper moment matching method. (a) and (b) correspond to the synthetic dataset. (c) and (d) are for the MNIST dataset. }\vspace{-0.5cm}
        \label{fig:histogram}
    \end{figure*}

    \subsection{Empirical Prior Distributions of Local Gradients}\label{distribution}
    The local gradients' prior distribution plays a key role when aggregating the heterogenous local gradient information across mobile devices. In this subsection, we numerically validate the accuracy of the approximation for the Gaussian and Laplacian prior models.   
    Fig. \ref{fig:histogram} shows the histograms for the elements of ${\bf g}_k^t$ and the corresponding Gaussian and Laplacian prior approximations with a proper moment matching method. In particular, we use the MLE estimator for the first and the second moments matching. As can be seen,  since the local data distribution is assumed to be heterogeneous across users, the users' prior distributions are also distinct. Also, one interesting observation is that the sample average of ${\bf g}_k^t$ is asymptotically zero, i.e., zero-mean regardless of the data distributions and the learning models. In contrast, the empirical distributions evolve distinctly over iterations depending on the users and the learning models. For instance, for  linear regression with synthetic heterogeneous datasets, the local gradients' variance dwindle over iterations. However, the local gradients' variance tends to increase when optimizing the CNN using the MNIST dataset. 
     
    \subsection{Learning Performance Comparison}\label{comparison}
    
                   \begin{figure*}[t]
    	\centering
        \includegraphics[width=0.9\linewidth]{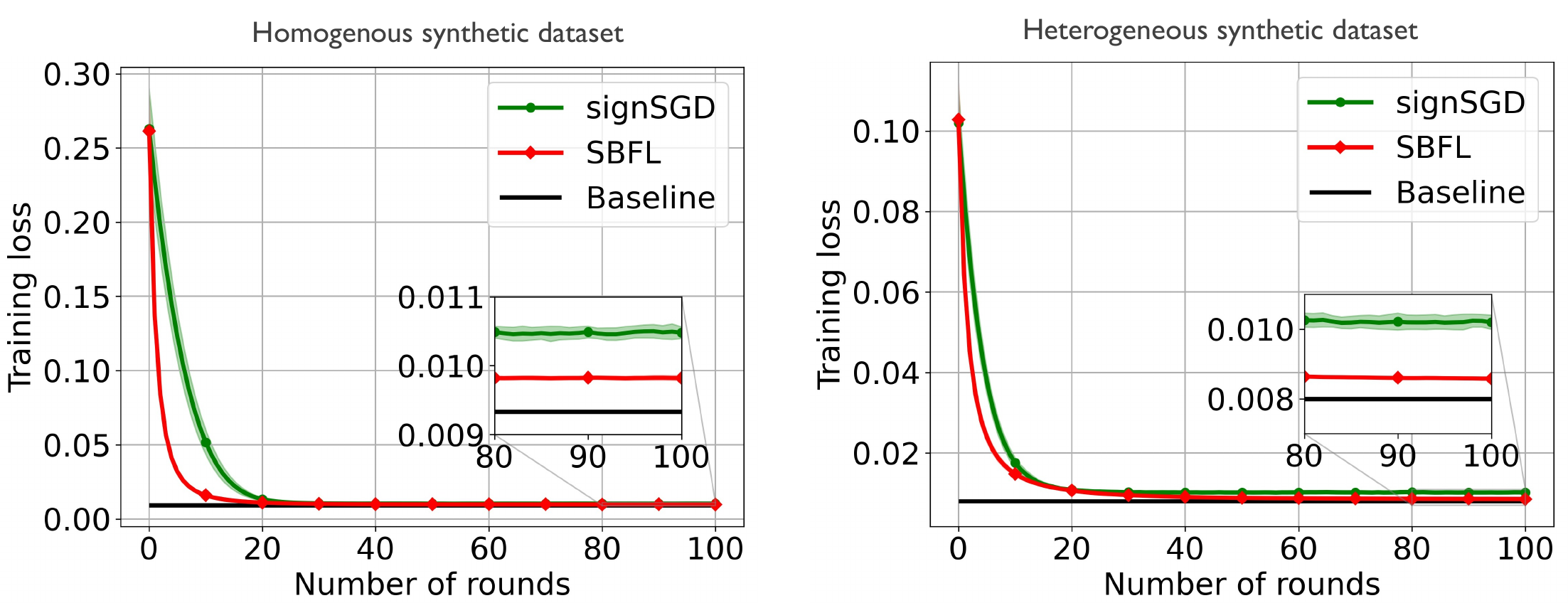}
      \caption{A comparison of training loss for homogeneous and heterogeneous synthetic datasets. Training losses after converged are shown in magnified box in each figure.} \label{fig:loss_syn}\vspace{-0.5cm}
    \end{figure*}
    
    
    \subsubsection{Synthetic dataset}
    Fig. \ref{fig:loss_syn} shows how the training losses diminishes over communication rounds when training the linear classifier using two synthetic datasets. For the homogeneous datasets with $a_k=5$ for $k\in [K]$, the proposed SBFL achieves training loss performance of 0.009, which is a 54$\%$ reduction with respect to that attained by the signSGD algorithm. This training loss gain is further magnified when considering the heterogenous dataset as seen in Fig. \ref{fig:loss_syn}. This gain stems from the joint exploitation of information about the distinct prior and channel distributions, which allows capturing the property of the data and channel link heterogeneities simultaneously. Another interesting observation is that the convergence speed also becomes faster by this synergetic gain. 
    
    \begin{table}[t]
    \centering
        \begin{threeparttable}
        \caption{ Learning efficiency comparison with signSGD for different hyper-parameters \label{Table:rounds}}
        \begin{tabular*}{\columnwidth}{@{\extracolsep{\fill}}cccccccc}
        \toprule 
        \multicolumn{2}{c}{Hyperparam.} & \multicolumn{6}{c}{Number of Rounds\tnote{*}} \\
        \cmidrule(){1-2} \cmidrule(lr){3-8}
        $\gamma$ & $\delta$ & \multicolumn{2}{c}{signSGD} & \multicolumn{2}{c}{SBFL + Gaussian} & \multicolumn{2}{c}{SBFL + Laplacian}\\
        \cmidrule(lr){3-4} \cmidrule(lr){5-6} \cmidrule(lr){7-8}
        & & TL & TA & TL & TA & TL & TA\\
        \midrule
        $10^{-2}$ & 0.0 & - & - & 62 & 170 & 92 & - \\
        $10^{-3}$ & 0.0 & \textbf{210} & - & 258 & 599 & 470 & 582 \\
        $10^{-2}$ & 0.9 & - & - & \textbf{42} & - & 90 & - \\
        $10^{-3}$ & 0.9 & - & \textbf{683} & 58 & \textbf{139} & \textbf{89} & \textbf{400} \\
        $10^{-4}$ & 0.9 & 243 & - & 265 & 627 & 494 & 557 \\
        \bottomrule
        \end{tabular*}
        
        \smallskip
        \scriptsize
        \begin{tablenotes}
        \item[*] The average number of communication rounds to reach a certain level of training loss (TL) and test accuracy (TA). 
        \end{tablenotes}
        \end{threeparttable}\vspace{-0.4cm}
    \end{table}


                  \begin{figure*}[t]
    	\centering
        \includegraphics[width=0.9\linewidth]{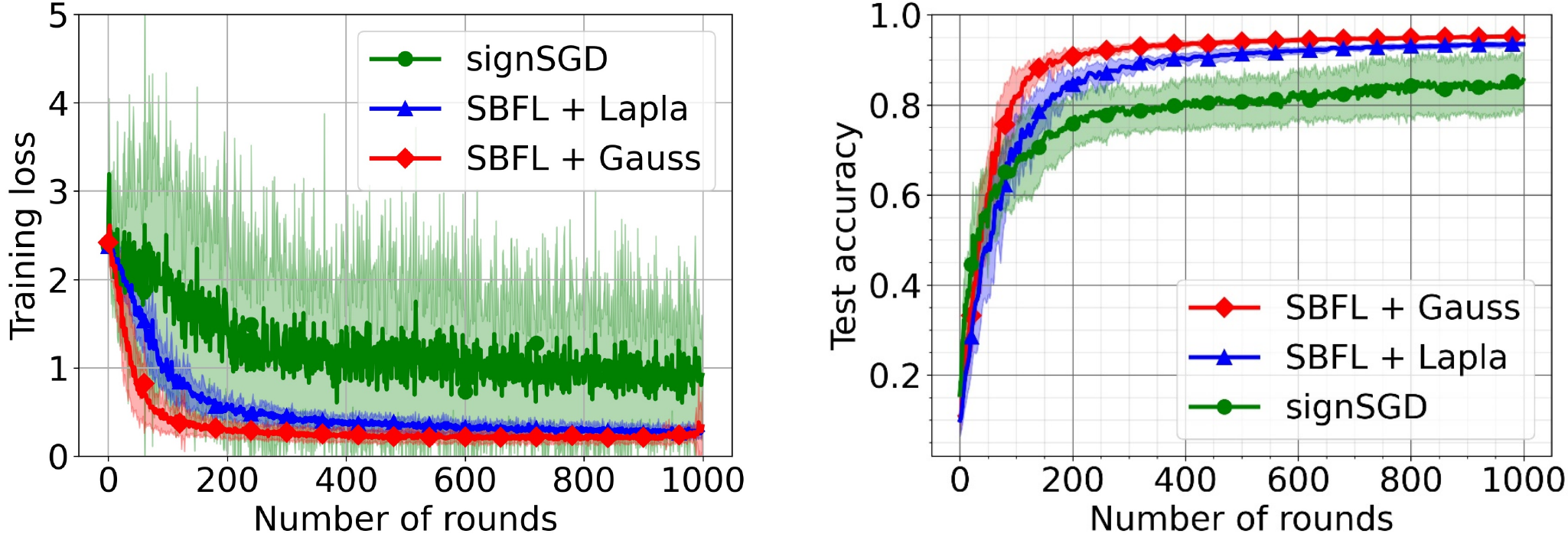}
      \caption{A comparison of test accuracy according to different learning rate $\gamma^t\in\{10^{-2},10^{-3}\}$.} \label{fig:noDC}\vspace{-0.3cm}
    \end{figure*}
    
    \subsubsection{MNIST dataset}
    Table \ref{Table:rounds} and \ref{Table:accuracy} summarize the simulation results varying hyper-parameters, including learning rates and weights for momentum. Since the magnitude of aggregated gradients varies according to the quantization and aggregation methods, the learning rates must be adjusted adequately for each method. In particular, in Table \ref{Table:rounds}, we evaluate the learning efficiencies for the different hyper-parameters in terms of the average number of communication rounds to reach a certain level of training loss or test accuracy. The levels of training loss (TA) and test accuracy (TA) are set to $1.0$ and $0.9$, respectively. Here, $-$ symbol denotes the method has not reached the level within $1000$ communication rounds. As can be seen,  when choosing hyper-parameter  $[\gamma,\delta]=[10^{-3},0.9]$, SBFL with Gaussian prior can reach the test accuracy of $0.9$ within 139 communication rounds in an average sense. In contrast, the sign SGD algorithm requires 638 communication rounds to meet the same level of test accuracy. Therefore, our learning algorithm approximately reduces the communication rounds by a factor of five. Meanwhile, SBFL with Laplacian prior takes more communication rounds than the Gaussian prior, which can be explained by the prior distribution model's mismatch effect with the true prior distribution. 
    
    Fig. \ref{fig:noDC} shows the training loss and the test accuracy of SBFL and signSGD algorithms with the best hyper-parameters for each algorithm in Table \ref{Table:rounds}. As can be seen, the proposed one provides significant gains compared to the signSGD algorithm in both the loss and accuracy. The proposed algorithm speeds up the convergence rates, saving communication costs for federated learning systems. Besides, SBFL with both prior distributions reduces the variance of training loss and test accuracy. This shows that SBFL can take advantage when aggregating the local gradients by jointly harnessing the link quality's heterogeneities and the prior distributions, even for optimizing a non-convex loss function.
   
   Fig \ref{fig:quantization} demonstrates how the test accuracy of SBFL changes according to quantization levels of $\nu_k^t$ for the Gaussian and $\lambda_k^t$ for the Laplacian priors. As can be seen, the test accuracy does not change even if  we quantize $\nu_k^t (\lambda_k^t)$ using a 4-bit uniform quantizer. This result verifies that SBFL can improve the learning performance with a very marginal additional communication cost for sending $\nu_k^t$ ($\lambda_k^t$).

                  \begin{figure*}[t]
    	\centering
        \includegraphics[width=0.9\linewidth]{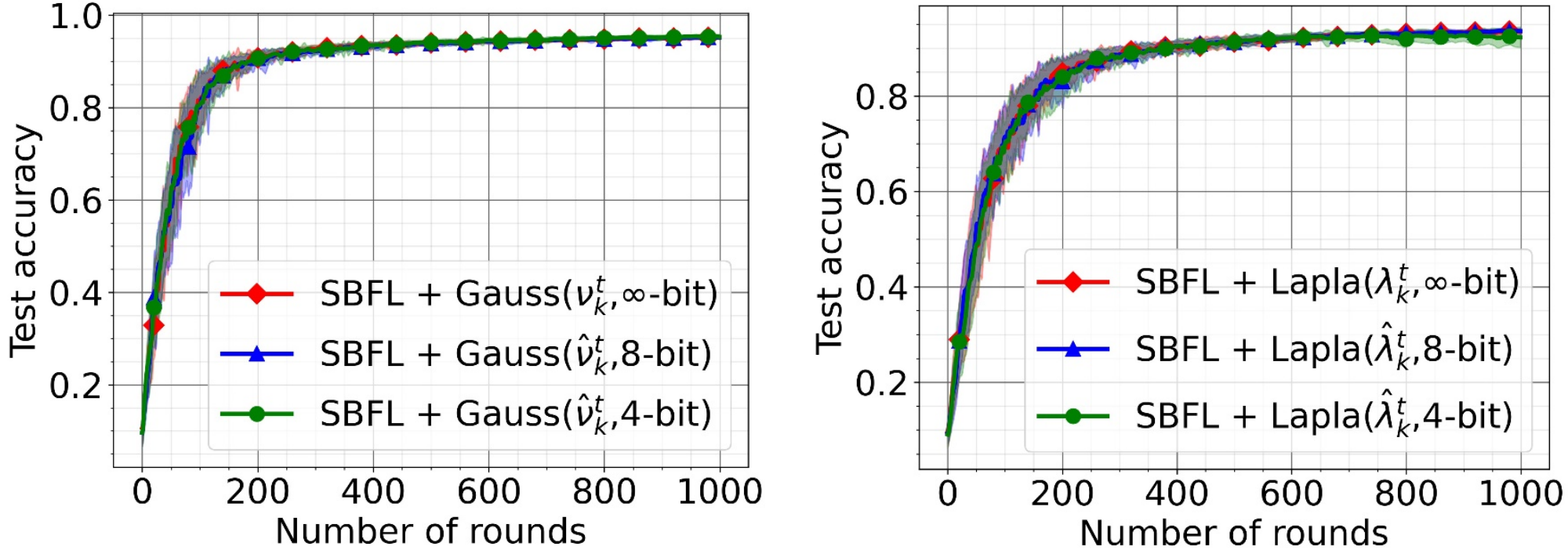}
      \caption{The effect of quantization levels of $\nu_k^t$ or ($\lambda_k^t$) according to SBFL with different priors.} \label{fig:quantization}\vspace{-0.5cm}
    \end{figure*}

    \begin{table}[t]
        \centering
        \begin{threeparttable}
        \caption{ Test accuracy comparison with signSGD for different hyper-parameters\label{Table:accuracy}}
        \begin{tabular*}{\columnwidth}{@{\extracolsep{\fill}}cccccccc}
        \toprule 
        \multicolumn{2}{c}{Hyperparam.} & \multicolumn{6}{c}{Test Accuracy\tnote{*}} \\
        \cmidrule(){1-2} \cmidrule(lr){3-8}
        $\gamma$ & $\delta$ & \multicolumn{2}{c}{signSGD} & \multicolumn{2}{c}{SBFL + Gaussian} & \multicolumn{2}{c}{SBFL + Laplacian}\\
        \cmidrule(lr){3-4} \cmidrule(lr){5-6} \cmidrule(lr){7-8}
        & & Full & DC & Full & DC & Full & DC\\
        \midrule
        $10^{-2}$ & 0.0 & 0.78 & 0.84 & \textbf{0.95} & 0.92 & \textbf{0.94} & 0.41 \\
        $10^{-3}$ & 0.0 & 0.86 & 0.77 & 0.89 & \textbf{0.94} & 0.81 & \textbf{0.94} \\
        $10^{-4}$ & 0.0 & 0.71 & 0.51 & 0.54 & 0.73 & 0.36 & 0.73 \\
        $10^{-2}$ & 0.9 & 0.11 & 0.35 & 0.94 & 0.10 & 0.71 & 0.10 \\
        $10^{-3}$ & 0.9 & \textbf{0.90} & \textbf{0.90} & \textbf{0.95} & 0.93 & \textbf{0.94} & 0.91 \\
        $10^{-4}$ & 0.9 & 0.84 & 0.78 & 0.89 & \textbf{0.94} & 0.82 & \textbf{0.94} \\
        \bottomrule
        \end{tabular*}
        
        \smallskip
        \scriptsize
        \end{threeparttable}
    \end{table}

                  \begin{figure*}[t]
    	\centering
        \includegraphics[width=0.9\linewidth]{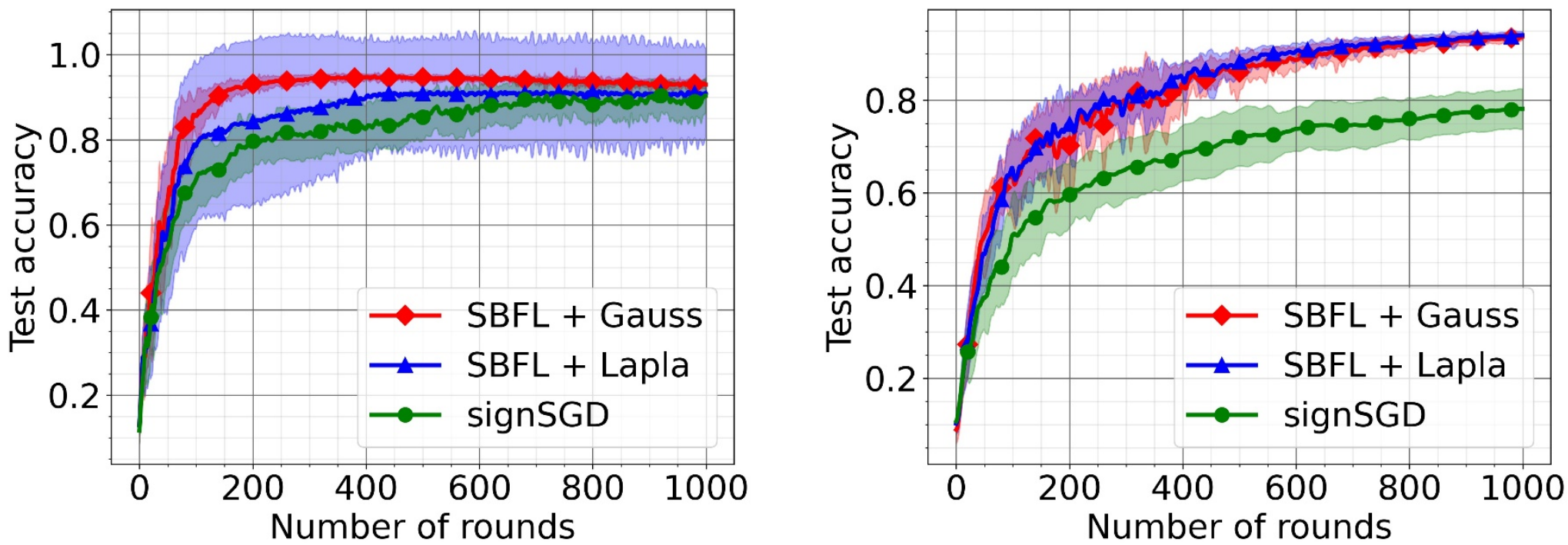}
      \caption{A comparison of test accuracy according to different learning rate $\gamma^t\in\{10^{-2},10^{-3}\}$ when the downlink compression is used.} \label{fig:DC}\vspace{-0.5cm}
    \end{figure*}
    

    In Table \ref{Table:accuracy}, we compare the test accuracies of the algorithms by altering learning hyper-parameters $\gamma\in \left\{10^{-2},10^{-3},10^{-4},10^{-5}\right\}$ and $\delta\in\{0,0.9\}$. The test accuracy values are obtained after $1000$ communication rounds. Compared to the no downlink compression case, in which the perfect model parameter ${\bf w}^t$ is delivered to mobile devices per communication round, we can observe that SBFL with Gaussian and Laplacian priors can enhance the test accuracy by $5\%$ and $4\%$ compared to the signSGD algorithm without downlink compression, respectively. To fairly compare the proposed algorithm with the benchmark signSGD, which only requires to send the binary sign vector in the downlink communication per round, element-wise sign quantization is applied on the downlink compression for each method. In this case, we observe that SBFL with downlink compression outperforms the corresponding signSGD algorithm in terms of test accuracy for various learning rates and weight parameters. As shown in Fig. \ref{fig:DC}, one interesting observation is that the test accuracy gain improves more than that of the case without the downlink compression. This elucidates that the proposed algorithm is robust to downlink compression; thereby, it is promising for wireless federated learning systems even when the downlink channel capacity is limited.

    

\section{Conclusion}
This work proposed a new Bayesian approach for federated learning over heterogeneous wireless networks, called BFL, and its communication-and-computation efficient variation, SBFL. We demonstrated that SBFL improves the training loss and test accuracy performance compared to signSGD with heterogeneous data and distinct wireless link qualities across mobile devices.  These performance gains are attainable by jointly harnessing the side-information on the users' local gradient priors and the channel distributions, facilitating aggregation of the local gradients more accurately per communication round. We showed  that for non-convex and smooth objectives, the models trained using SBFL with heterogeneous datasets converge to the optimal value. In simulations, we also demonstrated the ability of SBFL to learn complicated convolutional network models more accurately than signSGD when using non-synthetic datasets. 

Promising future research directions include investigating of the effect of over-the-air computation when sharing wireless links across mobile devices. Generalization to federated multi-task learning is also an interesting direction.  

\appendix
    \subsection{Proof for Proposition 1}\label{proof1}
To prove Proposition 1,  we need to compute the expectation of ${\bar g}_{k,m}^t$ conditioned on $ {y}_{k,m}^t$:
    \begin{align}
        \mathbb{E}\left[{ \bar g}_{k,m}^t\mid {y}_{k,m}^t\right]  =	\frac{\int_{-\infty}^{\infty}{\bar g}^t_{k,m}   P\left({y}_{k,m}^t|{\bar g}^t_{k,m}\right)P\left( {\bar g}^t_{k,m}\right) {\rm d}{\bar g}^t_{k,m}}{ \int_{-\infty}^{\infty}   P\left({y}_{k,m}^t|{\bar g}^t_{k,m}\right)P\left( {\bar g}^t_{k,m}\right) {\rm d}{\bar g}^t_{k,m}  }.
        \label{eq:scalemse} 
    \end{align}      	
    The numerator term in \eqref{eq:scalemse} is computed as
    \begin{align}
        \!\! \int_{-\infty}^{\infty} \!\frac{{\bar g}^t_{k,m}}{ 2\pi \sigma_k \nu_k^t}  \exp\left( -\frac{\left|y_{k,m}^t\!-\!h_k^t\!{\sf sign}\left( {\bar g}_{k,m}^t \right) \right|^2}{2\sigma_k^2}\! -\!\frac{|{\bar g}_{k,m}^t|^2}{2(\nu_k^t)^2} \right)  {\rm d}{\bar g}^t_{k,m}\!\!=\nu_k^t\frac{e^{- \frac{\left(h_k\!-\!y_{k,m}^t\right)^2}{2\sigma_k^2}}- e^{- \frac{\left(h_k+y_{k,m}^t\right)^2}{2\sigma_k^2}} }{2\pi \sigma_k }. \label{eq:numerator}
    \end{align}
    Similarly,  the denominator term in \eqref{eq:scalemse} is computed as
    \begin{align}
        \!\! \int_{-\infty}^{\infty} \!\frac{1}{ 2\pi \sigma_k \nu_k^t}\!\! \exp\left( -\frac{\left|y_{k,m}^t\!-\!h_k^t{\sf sign}\left({\bar g}_{k,m}^t\right) \right|^2}{2\sigma_k^2}\! -\!\frac{|{\bar g}_{k,m}^t |^2}{2(\nu_k^t)^2} \right)  {\rm d}{\bar g}^t_{k,m}=\frac{e^{-\frac{\left(h_k-y_{k,m}^t\right)^2}{2\sigma_k^2}} +e^{-\frac{\left(h_k+y_{k,m}^t\right)^2}{2\sigma_k^2}}}{2\sqrt{2\pi }\sigma_k}.\label{eq:denominator}
    \end{align}
    Invoking \eqref{eq:numerator} and \eqref{eq:denominator} into \eqref{eq:scalemse}, we obtain
    \begin{align}
       \mathbb{E}\left[{ \bar g}_{k,m}^t\mid { y}_{k,m}^t\right] &=\nu_k^t\sqrt{\frac{2}{\pi}} \frac{ \exp\left(\frac{2h_k^t y^t_{k,m}}{\sigma_k^2}\right)-\exp\left(-\frac{2h_k^t y^t_{k,m}}{\sigma_k^2}\right)}{\exp\left(\frac{2h_k^t y^t_{k,m}}{\sigma_k^2}\right)+\exp\left(-\frac{2h_k^t y^t_{k,m}}{\sigma_k^2}\right)} =\nu_k^t\sqrt{\frac{2}{\pi}} \tanh\left(\frac{2h_k^t y^t_{k,m}}{\sigma_k^2}\right), 
    \end{align} 	
which completes the proof.

%

    \subsection{Proof for Theorem 1}\label{proof2}
    To compute MSE, we need to calculate the conditional variance as
    \begin{align}
    	\!\mathbb{E}\!\left[ \!\left({ \bar g}_{k,m}^t\! -\! \mathbb{E}\left[{\bar g}_{k,m}^t\mid {y}_{k,m}^t\!\right]\!\right)^2 \right] =\mathbb{E}\left[{\sf Var}\left( { \bar g}_{k,m}^t\mid {  y}_{k,m}^t\right)\right] ={\sf Var}\left({ g}_{k,m}^t\right) - {\sf Var}\left(\mathbb{E}\left[{  \bar g}_{k,m}^t\mid { y}_{k,m}^t\right]\right). \label{eq:convar_mse}
    \end{align} 
Since ${\sf Var}\left({ g}_{k,m}^t\right)=(\nu_k^t)^2$, we only need to compute ${\sf Var}\left(\mathbb{E}\left[{\bar g}_{k,m}^t\mid { y}_{k,m}^t\right]\right)=\mathbb{E}\left[(\mathbb{E}\left[{\bar g}_{k,m}^t\mid { y}_{k,m}^t\right])^2\right] - \left(\mathbb{E}\left[\mathbb{E}\left[{\bar g}_{k,m}^t\mid { y}_{k,m}^t\right]\right]\right)^2 $. Since $ \mathbb{E}\left[{ \bar g}_{k,m}^t\mid { y}_{k,m}^t\right] =\nu_k^t\sqrt{\frac{2}{\pi}} \tanh\left(\frac{2h_k^t y^t_{k,m}}{\sigma_k^2}\right)$, we have
    \begin{align}
        \mathbb{E}_{{ y}_{k,m}^t}\left[ \nu_k^t\sqrt{\frac{2}{\pi}}\tanh\left(\frac{2h_k^ty_k^t}{\sigma_k^2}\right)	\right] 
        =\nu_k^t\sqrt{\frac{2}{\pi}} \int_{-\infty}^{\infty}\tanh\left(\frac{2h_k^ty_k^t}{\sigma_k^2}\right) P({y}_{k,m}^t){\rm d}y_{k,m}^t =0,
    \end{align}
    where $P({y}_{k,m}^t)$ is defined in \eqref{eq:py}.
   In addition, we compute
    \begin{align}
        \mathbb{E}_{{y}_{k,m}^t}\left[  \left(\nu_k^t\sqrt{\frac{2}{\pi}} \tanh\left(\frac{2h_k^ty_k^t}{\sigma_k^2}\right)\right)^2	\right]= \left(\nu_k^t\right)^2\frac{2}{\pi} \int_{-\infty}^{\infty}  \tanh\left(\frac{2h_k^ty_k^t}{\sigma_k^2}\right)^2P({y}_{k,m}^t) {\rm d}y_{k,m}^t. 
        \label{eq:condvar}
    \end{align}
    Unfortunately, this integration does not have a closed-form expression. Therefore, by plugging \eqref{eq:condvar} into \eqref{eq:convar_mse}, we obtain the minimum MSE as
    \begin{align}
        \!\mathbb{E} \left[  \left({\bar g}_{k,m}^t - \mathbb{E}\left[{ \bar g}_{k,m}^t\mid { y}_{k,m}^t\!\right]\right)^2 \right]&=(\nu_k^t)^2-\left(\nu_k^t\right)^2\frac{2}{\pi} \int_{-\infty}^{\infty}  \tanh\left(\frac{2h_k^ty_k^t}{\sigma_k^2}\right)^2P({y}_{k,m}^t) {\rm d}y_{k,m}^t.  \label{eq:convar_mse2}
    \end{align}

   \subsection{Proof for Corollary 3} \label{proof3}
 
     When using $U^{\star}_{\sf BLMMSE}({\bf y}_1^t,\ldots, {\bf y}_K^t)$, the minimum MSE is given by         	\begin{align}
         		\eta_{\sf bmse}^t &= \sum_{k=1}^K\mathbb{E}\!\left[ \!\left({\bf \bar g}_{k}^t - U^{\star}_{\sf BLMMSE}({\bf y}_k^t)\right)^2 \right] \nonumber\\
    	&= M\sum_{k=1}^K{\sf Var}\left({\bar g}_{k,m}^t\right) - {\sf Var}\left( U^{\star}_{\sf BLMMSE}({y}_{k,m}^t)\right).
         	\end{align}
         Recall that when $\mu_k=0$ for $k\in [K]$, $  U^{\star}_{\sf BLMMSE}({ y}_{k,m}^t)=    \frac{\sqrt{\frac{2}{\pi}}h_k^t\nu_k^t}{ \frac{2}{\pi}(h_k^t)^2+\sigma_k^2}{y}_{k,m}^t	$. Therefore, 
         \begin{align}
         	{\sf Var}\left( U^{\star}_{\sf BLMMSE}({y}_{k,m}^t)\right)&= \frac{ \frac{2}{\pi} (h_k^t)^2(\nu_k^t)^2}{ \left\{\frac{2}{\pi} (h_k^t)^2+\sigma_k^2\right\}^2}\mathbb{E}\left[({y}_{k,m}^t)^2\right]-  \frac{\sqrt{\frac{2}{\pi}}h_k^t\nu_k^t}{ \frac{2}{\pi}(h_k^t)^2+\sigma_k^2}\mathbb{E}\left[ {y}_{k,m}^t	\right]^2\nonumber\\
         	&=\frac{ \frac{2}{\pi} (h_k^t)^2(\nu_k^t)^2}{ \left\{\frac{2}{\pi} (h_k^t)^2+\sigma_k^2\right\}^2}\left\{\frac{2}{\pi} (h_k^t)^2+\sigma_k^2\right\}.
         \end{align}
         Since ${\sf Var}\left({\bar g}_{k,m}^t\right) =(\nu_k^t)^2$, the resultant MSE is given by
         \begin{align}
         	\eta_{\sf bmse}^t =M\sum_{k=1}^K (\nu_k^t)^2\left[1- \frac{ \frac{2}{\pi} (h_k^t)^2}{  \frac{2}{\pi} (h_k^t)^2+\sigma_k^2}\right].
         \end{align}

    \subsection{Proof for Theorem 2}\label{proof4}

We prove Theorem 2  by relating the norm of the aggregated gradient to the expected improvement per communication round, comparing it with the total possible improvement, which is a widely-adopted strategy in the convergence anlaysis.

Let ${\bf \hat g}^t_{\Sigma}=U_{\sf MMSE}^{\star}({\bf y}_1^t,\ldots, {\bf y}_K^t)$ and ${\bf \hat g}^t_{\Sigma}={\bf g}^t_{\Sigma} +{\bf e}^t$ with $\mathbb{E}\left[\|{\bf e}^t\|_2^2\right] \leq \sigma^2_{\sf mse}$. From Assumption 2,  we compute the objective function improvement in a single algorithmic step as
    \begin{align}
    	 f\left({\bf w}^{t+1}\right) -  f\left({\bf w}^t\right) &\leq  \left({\bf g}_{\Sigma}^t\right)^{\top}\left( {\bf w}^{t+1}-{\bf w}^{t}\right) + \frac{L}{2}\left\|{\bf w}^{t+1}-{\bf w}^t\right\|_2^2 \nonumber\\
    	 &=-\gamma^t\left({\bf g}^t_{\Sigma}\right)^{\top}{\bf \hat g}^t_{\Sigma} +\left(\gamma^t\right)^2\frac{L}{2}\|{\bf \hat g}^t_{\Sigma}\|_2^2\nonumber\\
    	 &=-\gamma^t\left({\bf g}^t_{\Sigma}\right)^{\top}\left({\bf g}^t_{\Sigma} +{\bf e}^t\right)+\left(\gamma^t\right)^2\frac{L}{2}\|{\bf g}^t_{\Sigma}+{\bf e}^t\|_2^2.
    \end{align}
    Since $\mathbb{E}\left[\left({\bf g}^t_{\Sigma}\right)^{\top}{\bf e}^t\mid {\bf w}^t\right]=0$ and $\mathbb{E}\left[\|{\bf e}^t\|_2^2\right]\leq \sigma^2_{\sf mse}$, we can decompose the mean squared stochastic gradient into its mean and variance by taking the expectation conditioned on previous iterations as
     \begin{align}
    	\mathbb{E}\left[ f\left({\bf w}^{t+1}\right) -  f\left({\bf w}^t\right) \mid {\bf w}^t \right]\leq -\gamma^t\|{\bf  g}^t_{\Sigma}\|_2^2 +\left(\gamma^t\right)^2\frac{L}{2}\left(\|{\bf g}^t_{\Sigma}\|_2^2 +\sigma^2_{\sf mse}\right).\label{eq:condexp}
    \end{align}
   Invoking the adaptive learning rate $\gamma^t =\frac{1}{t+1} \leq \frac{1}{\sqrt{t+1}}$ and plugging it into \eqref{eq:condexp} we obtain
       \begin{align}
    	\mathbb{E}\left[ f\left({\bf w}^{t+1}\right) -  f\left({\bf w}^t\right) \mid {\bf w}^t \right]&\leq -\frac{\gamma}{\sqrt{t+1}}\|{\bf  g}^t_{\Sigma}\|_2^2 +\frac{\gamma^2}{t+1}\frac{L}{2}\left(\|{\bf g}^t_{\Sigma}\|_2^2 +\sigma^2_{\sf mse}\right) \nonumber \\
    	&\leq -\frac{\gamma}{\sqrt{t+1}}\|{\bf  g}^t_{\Sigma}\|_2^2\left(1-\frac{\gamma L}{2}\right) +\frac{\gamma^2}{t+1}\frac{L}{2} \sigma^2_{\sf mse}.
    \end{align}  
Taking the expectation over ${\bf w}^t$ and applying the method of telescoping sums over $t\in [T]$, we obtain
\begin{align}
	 f\left({\bf w}^{0}\right) -  f\left({\bf w}^{\star}\right) &\geq \mathbb{E}\left[\sum_{t=0}^{T-1}f\left({\bf w}^{t}\right)-f\left({\bf w}^{t+1}\right)\right]\nonumber\\
	 &\geq \sum_{t=0}^{T-1}\left[\frac{\gamma}{\sqrt{t+1}} \mathbb{E}\left[\|{\bf g}_{\Sigma}^t\|_2^2\right] \left(1-\frac{\gamma L}{2}\right)  -\frac{\gamma^2}{t+1}\frac{L}{2} \sigma^2_{\sf mse}\right]\nonumber\\
	 	 &\geq T\left[\frac{\gamma}{\sqrt{T}} \mathbb{E}\left[ \frac{1}{T}\sum_{t=0}^{T-1} \|{\bf g}_{\Sigma}^t\|_2^2\right] \left(1-\frac{\gamma L}{2}\right)\right]  -\sum_{t=0}^{T-1}\left[\frac{\gamma^2}{t+1}\frac{L}{2} \sigma^2_{\sf mse}\right]\nonumber\\
	 	 	 &\geq  \sqrt{T}\gamma \mathbb{E}\left[ \frac{1}{T}\sum_{t=0}^{T-1} \|{\bf g}_{\Sigma}^t\|_2^2\right] \left(1-\frac{\gamma L}{2}\right)  - \left(1+\ln T\right)\gamma^2 \frac{L}{2} \sigma^2_{\sf mse},
\end{align}   
where the last inequality is due to the harmonic sum, i.e., $\sum_{t=0}^{T-1}\frac{1}{1+t}\leq 1+\ln (T)$. This completes the proof.

    \bibliographystyle{IEEEtran}
    \bibliography{BayesianAggregation}

\end{document}